\documentclass[a4paper]{amsart}
\usepackage{dsfont}
\usepackage{algorithm}
\usepackage{algpseudocode}
\usepackage[numbers]{natbib}
\usepackage{graphicx}
\usepackage{amsthm}
\usepackage{booktabs}
\usepackage[shortlabels]{enumitem}
\usepackage{verbatim}
\usepackage[utf8]{inputenc}
\usepackage[english]{babel}
\usepackage{color}
\usepackage{amsmath, amssymb}
\usepackage[margin=3cm]{geometry}
\usepackage{url}

\newcommand{\Binom}{\mathrm{Binomial}} 
\newcommand{\NegBinom}{\mathrm{NegativeBinomial}} 
\newcommand{\Normal}{\mathrm{Normal}} 
\newcommand{\Categ}{\mathrm{Categ}} 

\newcommand{\ra}[1]{\renewcommand{\arraystretch}{#1}} 

\newcommand{\ud}{\mathrm{d}} 
\newcommand{\uarg}{\,\cdot\,} 
\newcommand{\neff}{\mathrm{n}_{\mathrm{eff}}} 
\newcommand{\iact}{\mathrm{IACT}} 
\newcommand{\ire}{\mathrm{IRE}} 
\newcommand{\refr}[1]{{\dot{#1}}} 
\newcommand{\logit}{\mathrm{logit}} 
\newcommand{\alphatarget}{\alpha_*}
\newcommand{\alphaopt}{\alpha_{\mathrm{opt}}} 
\newcommand{\adaptstate}{\zeta} 
\newcommand{\adaptspace}{\mathsf{Z}} 
\newcommand{\adaptdata}{\xi} 
\newcommand{\X}{\mathsf{X}} 
\newcommand{\R}{\mathbb{R}} 
\newcommand{\cnQ}[1]{Q_{#1}^{\mathrm{AR}}} 
\newcommand{\rwQ}[1]{Q_{#1}^{\mathrm{RW}}} 
\newcommand{\rwq}[1]{q_{#1}^{\mathrm{RW}}} 

\newcommand{\rnaught}{\mathcal{R}_0} 
\newcommand{\popsize}{N_{\mathrm{pop}}} 
\graphicspath{{./plots/}}

\newtheorem{theorem}{Theorem}

\newtheorem{lemma}[theorem]{Lemma}

\theoremstyle{remark}

\theoremstyle{definition}
\newtheorem{definition}{Definition}
\newtheorem{assumption}{Assumption}

\title{Conditional particle filters with diffuse initial distributions}
\author{Santeri Karppinen}
\author{Matti Vihola}

\begin{document}

\begin{abstract}
Conditional particle filters (CPFs) are powerful smoothing algorithms
for general nonlinear/non-Gaussian hidden Markov models. However, CPFs
can be inefficient or difficult to apply with diffuse initial
distributions, which are common in statistical applications. We
propose a simple but generally applicable auxiliary variable method,
which can be used together with the CPF in order to perform efficient
inference with diffuse initial distributions. The method only requires
simulatable Markov transitions that are reversible with respect to the
initial distribution, which can be improper. We focus in particular on
random-walk type transitions which are reversible with respect to a
uniform initial distribution (on some domain), and autoregressive
kernels for Gaussian initial distributions. We propose to use on-line
adaptations within the methods. In the case of random-walk transition,
our adaptations use the estimated covariance and acceptance rate
adaptation, and we detail their theoretical validity. We tested our
methods with a linear-Gaussian random-walk model, a stochastic
volatility model, and a stochastic epidemic compartment model with
time-varying transmission rate. The experimental findings demonstrate
that our method works reliably with little user specification, and can be
substantially better mixing than a direct particle Gibbs algorithm
that treats initial states as parameters.
\end{abstract}



\maketitle

\section{Introduction} 

In statistical applications of general state space hidden Markov
models (HMMs), commonly known also as state space models, it is often
desirable to initialise the latent state of the model with a diffuse
(uninformative) initial distribution
\citep[cf.][]{durbin-koopman2012}. We mean by `diffuse' the general
scenario, where the first marginal of the smoothing distribution is
highly concentrated relative to the prior of the latent
Markov chain, which may also be improper.

The conditional particle filter (CPF) \citep{pmcmc} and
in particular its backward sampling variants
\citep{backwardsampling,lindsten-jordan-schon} have been found to
provide efficient smoothing even with long data records, both
empirically \citep[e.g.][]{fearnhead2018particle} and theoretically
\citep{lee-singh-vihola}. However, a direct application of the CPF
to a model with a diffuse initial distribution will lead to poor
performance, because most of the initial particles will ultimately be redundant, as
they become drawn from highly unlikely regions of the state space.

There are a number of existing methods which can be used to mitigate
this inefficiency. For simpler settings, it is often relatively
straightforward to design proposal distributions that lead to
an equivalent model, which no longer has a diffuse
initial distribution. Indeed, if the first filtering distribution is
already informative, its analytical approximation may be used directly
as the first proposal distribution. The iteratively refined look-ahead
approach as suggested in \citep{guarniero-johansen-lee} extends to
more complicated settings, but can require careful tuning for
each class of problems.

We aim here for a general approach, which does not rely on any
problem-specific constructions. Such a general approach which allows
for diffuse initial conditions with particle Markov chain Monte Carlo
(MCMC) is to include the initial latent state of the HMM as a
`parameter'.  This was suggested in \citep{murray_disturbance} with
the particle marginal Metropolis-Hastings (PMMH).  The same approach
is directly applicable also with the CPF (using particle Gibbs); see
\cite{fearnhead_augmentation} who discuss general approaches based on
augmentation schemes.

Our approach may be seen as an instance of the general
`pseudo-observation' framework of \cite{fearnhead_augmentation}, but we
are unaware of earlier works about the specific class of methods we
focus on here. Indeed, instead of building the auxiliary variable from
the conjugacy perspective as in \cite{fearnhead_augmentation}, our
approach is based on Markov transitions that are reversible with
respect to the initial measure of the HMM. This approach may be
simpler to understand and implement in practice, and is very generally
applicable. We focus here on two concrete cases: the `diffuse
Gaussian` case, where the initial distribution is Gaussian with a relatively
uninformative covariance matrix, and the `fully diffuse` case, where
the initial distribution is uniform. We suggest on-line adaptation mechanisms
for the parameters, which make the methods easy to apply in practice.

We start in Section \ref{sec:aux} by describing the family of models
we are concerned with, and the general auxiliary variable
initialisation CPF that underlies all of our developments. We
present the practical methods in Section \ref{sec:diffuse-init}.
Section \ref{sec:experiments} reports
experiments of the methods with three academic models, and concludes
with a realistic inference task related to
modelling the COVID-19 epidemic in Finland. We conclude with
a discussion in Section \ref{sec:discussion}.

\section{The model and auxiliary variables}
\label{sec:aux} 

Our main interest is with HMMs having a joint smoothing distribution
$\pi$ of the following form:
\begin{equation}
  \label{eq:general_ssm}
  \pi(x_{1:T}) \propto p(x_1)p(y_1\mid x_1)\prod_{k=2}^{T} p(x_k \mid x_{k-1})p(y_k \mid x_k),
\end{equation}
where $\ell$:$u$ denotes the sequence of integers from $\ell$ to $u$
(inclusive), $x_{1:T}$ denotes the latent state variables, and
$y_{1:T}$ the observations. Additionally, $\pi$ may depend on
(hyper)parameters $\theta$, the dependence on which we omit for now,
but return to later, in Section \ref{sec:with-pg}.

For the convenience of notation, and to allow for some
generalisations, we focus on the Feynman-Kac form of the HMM smoothing
problem \cite[cf.][]{del-moral}, where the
distribution of interest $\pi$ is represented in terms of a
$\sigma$-finite measure $M_1(\ud x_1)$ on the state space $\X$,
Markov transitions $M_2,\ldots,M_T$ on $\X$ and
potential functions $G_k:\X^k\to[0,\infty)$ so that
\begin{equation}
  \label{eq:feynman_kac}
  \pi(\ud x_{1:T}) \propto M_1(\ud x_1)G_1(x_1) \prod_{k = 2}^{T}
  M_k(x_{k-1}, \ud x_k )G_k(x_{1:k}).
\end{equation}
The classical choice, the so-called `bootstrap filter' \citep{gordon-salmond-smith}, 
corresponds to
$M_1(\ud x_1) = p(x_1) \ud x_1$ and
$M_k(x_{k-1}, \ud x_k) = p(x_k\mid x_{k-1}) \ud x_k$, where `$\ud x$'
stands for the Lebesgue measure on $\X=\R^d$, and $G_k(x_{1:k}) =
p(y_k\mid x_k)$, but other choices with other `proposal distributions'
$M_k$ are also possible.
Our main focus is when $M_1$ is diffuse
with respect to the first marginal of $\pi$.
We stress that our method accomodates also improper $M_1$, such
as the uniform distribution on $\mathbb{R}^d$, as long as
\eqref{eq:feynman_kac} defines a probability.

The key ingredient of our method is an auxiliary Markov transition,
$Q$, which we can simulate from, and which satisfies the following:
\begin{assumption}[$M_1$-reversibility]
The Markov transition probability $Q$ is reversible with respect to the
$\sigma$-finite measure $M_1$, or $M_1$-reversible, if
\begin{equation}
  \label{eq:detailed_balance}
  \int M_1(\ud x_0)Q(x_0, \ud x_1) \mathbf{1}(x_0\in A,x_1\in B)
  = \int M_1(\ud x_1)Q(x_1, \ud x_0) \mathbf{1}(x_0\in A,x_1\in B),
\end{equation}
for all measurable $A,B\subset\X$.
\end{assumption}
We discuss practical ways to choose $Q$ in Section \ref{sec:diffuse-init}.
Assuming an $M_1$-reversible $Q$, we
define an augmented
target distribution, involving a new `pseudo-state' $x_0$ which is
connected to $x_1$ by $Q$:
\begin{equation*}
  \tilde{\pi}(\ud x_{0:T}) = \pi(\ud x_{1:T})Q(x_1, \ud x_0)
  \propto M_1(\ud x_0)Q(x_0, \ud x_1)
  G_1(x_1) \prod_{k = 2}^{T} M_k(x_{k-1}, \ud x_k)
                       G_k(x_{1:k}).
\end{equation*}
It is clear by construction that $\tilde{\pi}$ admits $\pi$ as its
marginal, and therefore, if we can sample $x_{0:T}$ from
$\tilde{\pi}$, then $x_{1:T}\sim \pi$.

Our method may be viewed as a particle Gibbs
\citep{pmcmc} which targets
$\tilde{\pi}$, regarding $x_0$ as the `parameter,' and $x_{1:T}$ the
`latent state,'
which are updated using the CPF. Algorithm \ref{alg:ai-cpf}
summarises the method, which we call the `auxiliary initialisation' CPF (AI-CPF).
Algorithm \ref{alg:ai-cpf} determines a $\pi$-invariant Markov transition
$\refr{x}_{1:T} \to \tilde{X}_{1:T}^{(B_{1:T})}$; the latter output of the
algorithm will be relevant later, when we discuss adaptation.

\begin{algorithm}
    \caption{$\text{AI-CPF}(\refr{x}_{1:T}; Q, M_{2:T}, G_{1:T}, N)$}
     \label{alg:ai-cpf} 
\begin{algorithmic}[1]
\State \label{dcpf:gibbs} Simulate $X_0 \sim Q(\refr{x}_1, \uarg)$.
\State \label{dcpf:newinit} Simulate $\tilde{X}_1^{(2:N)} \sim Q(X_0, \uarg)$ and set
$\tilde{X}_1^{(1)} = \refr{x}_1$.
\State \label{dcpf:f-cpf} Run
$(\tilde{X}_{1:T}^{(1:N)},W_{1:T}^{(1:N)},A_{1:T-1}^{(1:N)})
\gets \text{F-CPF}(\refr{x}_{2:T},
\tilde{X}_1^{(1:N)}; M_{2:T}, G_{1:T}, N) $.
\State \label{dcpf:cpf-select} 
$(B_{1:T}, V^{(1:N)}) \gets
\textsc{PickPath-x}(\tilde{X}_{1:T}^{(1:N)},W_{1:T}^{(1:N)},
A_{1:T-1}^{(1:N)}, M_{2:T},
G_{2:T})$.
\State \textbf{output} $\big(\tilde{X}_{1:T}^{(B_{1:T})}, (B_{1},
V^{(1:N)}, \tilde{X}_1^{(1:N)})\big),\quad$ where $\tilde{X}_{1:T}^{(B_{1:T})} =
    \big(\tilde{X}_{1}^{(B_1)},\tilde{X}_2^{(B_2)},\ldots,
        \tilde{X}_T^{(B_T)}\big)$.
\end{algorithmic}
\end{algorithm}

Line \ref{dcpf:gibbs}
of Algorithm \ref{alg:ai-cpf}
implements a Gibbs step sampling $X_0$ conditional on
$X_{1:T}=\refr{x}_{1:T}$, and lines
\ref{dcpf:newinit}--\ref{dcpf:cpf-select} implement together a CPF targetting
the conditional of $X_{1:T}$ given $X_0$. Line \ref{dcpf:f-cpf} runs
what we call a `forward' CPF, which is just a standard CPF conditional
  on the first state particles $X_{1}^{(1:N)}$,
detailed in Algorithm \ref{alg:f-cpf}. Line \ref{dcpf:cpf-select}
refers to a call of $\textsc{PickPath-AT}$ (Algorithm
\ref{alg:ancestortracing}) for
ancestor tracing as in the original work \citep{pmcmc},
or $\textsc{PickPath-BS}$ (Algorithm \ref{alg:backwardsampling}) for backward sampling
\citep{backwardsampling}. $\mathrm{Categ}(w^{(1:N)})$ stands for
the categorical distribution, that is, $A \sim \mathrm{Categ}(w^{(1:N)})$
if $\Pr(A=i) = w^{(i)}$.

\begin{algorithm}
  \caption{$\text{F-CPF}(\refr{x}_{2:T}, X_{1}^{(1:N)}; M_{2:T}, G_{1:T}, N)$}
  \label{alg:f-cpf} 
\begin{algorithmic}[1]
\State Set $\mathbf{X}_1^{(1:N)} \gets X_1^{(1:N)}$.
\For{$k = 1, \ldots, T - 1$}
\State Calculate $\tilde{W}_{k}^{(i)} \gets G_{k}(\mathbf{X}_{k}^{(i)})$
and $W_{k}^{(i)} \gets
\tilde{W}_k^{(i)}/\sum_{j=1}^N \tilde{W}_k^{(j)}$ for $i\in\{1{:}N\}$.
\State $A_{k}^{(2:N)} \sim \mathrm{Categ}\big(
W_{k}^{(1:N)}\big)$ and set $A_k^{(1)}\gets 1$.
\State Draw $X_{k+1}^{(i)} \sim M_{k+1}(\uarg \mid
            X_{k}^{(A_{k}^{(i)})})$ for
            $i \in \{2{:}N\}$
            and set $X_{k+1}^{(1)} = \refr{x}_{k+1}$.
\State Set $\mathbf{X}_{k+1}^{(i)} = (\mathbf{X}_k^{(A_{k}^{(i)})}, X_{k+1}^{(i)})$ for
$i \in \{1{:}N\}$.
\EndFor
\State Calculate $\tilde{W}_{T}^{(1:N)} \gets G_T(\mathbf{X}_T^{(1:N)})$
and $W_{T}^{(i)} \gets \tilde{W}_T^{(i)}/\sum_{j=1}^N \tilde{W}_T^{(j)}$
for $i=\{1{:}N\}$.
\State \textbf{output} $(X_{1:T}^{(1:N)}$, $W_{1:T}^{(1:N)}$,
$A_{1:T-1}^{(1:N)})$.
\end{algorithmic}
\end{algorithm} 
The ancestor tracing variant can be used when
the transition densities are unavailable. However, our main
interest here is with backward sampling, summarised in Algorithm
\ref{alg:backwardsampling} in the common case where the potentials only
depend on two consecutive states, that is, $G_k(x_{1:k}) =
G_k(x_{k-1:k})$, and the transitions admit densities $M_k(x_{k-1},\ud x_k) = M_k(x_{k-1},x_k)
\ud x_k$ with respect to some
dominating $\sigma$-finite measure `$\ud x_k$.'

\begin{algorithm}
  \caption{$\textsc{PickPath-AT}(\tilde{X}_{1:T}^{(1:N)},
    W_{1:T}^{(1:N)}, A_{1:T-1}^{(1:N)},
    M_{2:T}, G_{2:T})$}
  \label{alg:ancestortracing} 
  \begin{algorithmic}[1]
    \State Draw $B_K \sim \Categ\big(W_{T}^{(1:N)})$.
    \State \textbf{output} $(B_{1:T}, W_1^{(1:N)})$ where $B_k = A_k^{(B_{k+1})}$
      for $k=T-1,\ldots,1$.
  \end{algorithmic} 
\end{algorithm}

\begin{algorithm}
  \caption{$\textsc{PickPath-BS}(\tilde{X}_{1:T}^{(1:N)},
    W_{1:T}^{(1:N)}, A_{1:T-1}^{(1:N)},
    M_{2:T}, G_{2:T})$}
  \label{alg:backwardsampling} 
  \begin{algorithmic}[1]
    \State Draw $B_K \sim \Categ\big(W_{T}^{(1:N)})$.
    \For{$k = T - 1, \dots, 1$}
      \State Calculate $\tilde{V}_k^{(i)} \gets
        W_k^{(i)} M_{k + 1}(\tilde{X}_k^{(i)}, \tilde{X}_{k+1}^{(B_{k + 1})} )
        G_{k+1}(\tilde{X}_k^{(i)}, \tilde{X}_{k+1}^{(B_{k + 1})})$ for $i \in \{1{:}N\}$.
      \State Simulate $B_k \sim \Categ(V_k^{(1:N)})$, where $V_k^{(i)}
      = \tilde{V}_k^{(i)}/\sum_{j=1}^N \tilde{V}_k^{(j)}$.
    \EndFor
    \State \textbf{output} $(B_{1:T}, V_{1}^{(1:N)})$.
  \end{algorithmic} 
\end{algorithm}

We conclude with a brief discussion on the general method of Algorithm
\ref{alg:ai-cpf}.
\begin{enumerate}[(i)]
    \item We recognise that Algorithm \ref{alg:ai-cpf} is not
      new per se, in that it may be viewed just as a particle Gibbs
      applied for a specific auxiliary variable model. However, we are
      unaware of Algorithm \ref{alg:ai-cpf} being presented with the
      present focus: with an $M_1$-reversible $Q$, and allowing for
      an improper $M_1$.
    \item Algorithm \ref{alg:ai-cpf} may be viewed as a generalisation
      of the standard CPF. Indeed, taking $Q(x_0,\ud x_1) = M_1(\ud x_1)$
      in Algorithm \ref{alg:ai-cpf} leads to the standard CPF. Note
      that Line \ref{dcpf:gibbs} is redundant in this case, but is
      necessary in the general case.
    \item In the case $T=1$, Line \ref{dcpf:f-cpf} of
      Algorithm \ref{alg:ai-cpf} is redundant,
      and the algorithm resembles certain multiple-try Metropolis methods
      \cite[cf.][]{martino}, and has been
      suggested earlier \citep{mendes-scharth-kohn}.
    \item Algorithm \ref{alg:f-cpf} is formulated using multinomial
      resampling, for simplicity. We note that any other unbiased
      resampling may be used, as long as the conditional resampling is
      designed appropriately; see \cite{chopinsingh2015}.
\end{enumerate}
The `CPF generalisation' perspective of Algorithm \ref{alg:ai-cpf} may
lead to other useful developments; for instance, one could imagine the
approach to be useful with the CPF applied for static (non-HMM) targets,
as in sequential Monte Carlo samplers \citep{delmoral-doucet-jasra}. The
aim of the present paper is, however, to use Algorithm
\ref{alg:ai-cpf} with diffuse initial distributions.



\section{Methods for diffuse initialisation of conditional particle filters}
\label{sec:diffuse-init} 


To illustrate the typical problem that arises with a diffuse
initial distribution $M_1$, we examine a simple
noisy AR(1) model:
\begin{equation}
  \label{eq:noisyar}
  \begin{aligned}
    x_{k+1} &= \rho x_{k} + \eta_k,&  \eta_k \sim N(0, \sigma_x^2)\\
    y_{k} &= x_k + \epsilon_k, & \epsilon_k \sim N(0, \sigma_y^2),
  \end{aligned}
\end{equation}
for $k\ge 1$, $x_1 \sim N(0, \sigma_1^2)$, $M_1(\ud x_1) = p(x_1) \ud x_1$, 
$M_k(x_{k-1}, \ud x_k) = p(x_k\mid x_{k-1}) \ud x_k$ and $G_k(x_{1:k}) = p(y_k\mid x_k)$.

We simulated a dataset of length $T=50$ from this model
with $x_1 = 0$, $\rho = 0.8$ and $\sigma_x = \sigma_y = 0.5$.
We then ran 6000 iterations of the
CPF with backward sampling (CPF-BS) with $\sigma_1 \in \{10, 100, 1000\}$; that is,
Algorithm \ref{alg:ai-cpf} with $Q(x_0,\uarg) = M_1(\uarg)$ together with Algorithm
\ref{alg:backwardsampling},
and discarded the first 1000 iterations as burn-in.
For each value of $\sigma_1$, we monitored the efficiency of sampling $x_1$.
Figure \ref{fig:diffinit-poor-mixing-example} displays the resulting traceplots.
The estimated integrated autocorrelation times ($\iact$) were approximately
3.75, 28.92 and 136.64, leading to effective sample sizes ($\neff$) of
1600, 207 and 44, respectively.
This demonstrates how the performance of the CPF-BS deteriorates as the
initial distribution of the latent state becomes more diffuse.
\begin{figure}
  \centering
  \includegraphics[width=0.8\textwidth]{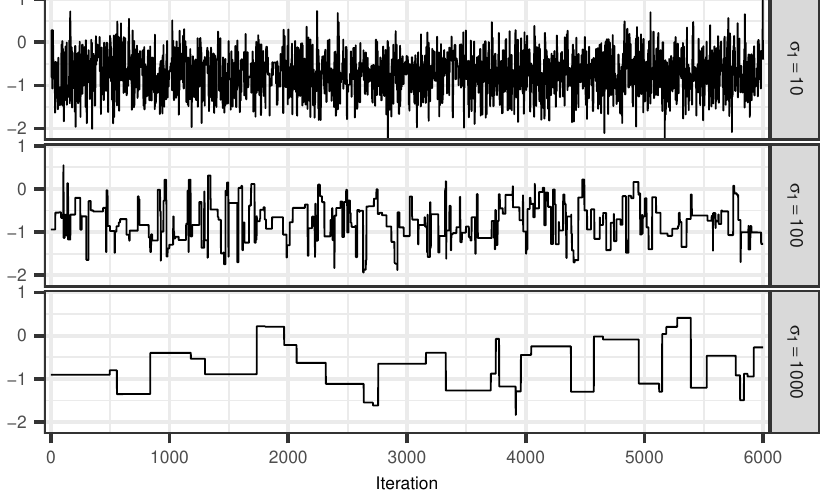}
  \caption{Traceplot of the initial state of the
    the noisy AR(1) model, using the CPF with 16 particles and
    backward sampling with $\sigma_1 = 10$ (top), $100$ (middle) and $1000$ (bottom).}
  \label{fig:diffinit-poor-mixing-example}
\end{figure}


\subsection{Diffuse Gaussian initialisation}
\label{sec:dgi-cpf} 
In the case that $M_1$ in \eqref{eq:feynman_kac} is Gaussian with
mean $\mu$ and covariance $\Sigma$,
we can construct a Markov transition function that satisfies
\eqref{eq:detailed_balance} using an autoregressive proposal similar to
`preconditioning' in the Crank-Nicolson algorithm
\citep[cf.][]{cotter-roberts-stuart-white}.
This proposal comes with a parameter $\beta \in (0, 1]$, so we denote this kernel
by $\cnQ{\beta}$.
A variate $Z \sim \cnQ{\beta}(x, \uarg)$ can be drawn simply by setting
\begin{equation}
  \label{eq:cn_sampling}
Z = \sqrt{1 - \beta^2}(x - \mu) + \beta W + \mu,
\end{equation}
where $W \sim N(0, \Sigma)$.
We refer to Algorithm \ref{alg:ai-cpf} with
$Q = \cnQ{\beta}$ as the diffuse Gaussian initialisation CPF (DGI-CPF).
In the special case $\beta = 1$, we have $\cnQ{1} = M_1$, and so
the DGI-CPF is equivalent with the standard CPF.


\subsection{Fully diffuse initialisation}
\label{sec:fdi-cpf} 
Suppose that $M_1(\ud x) = M_1(x) \ud x$ where $M_1(x)\equiv 1$ is a uniform
density on $\X=\R^d$. Then,
any symmetric transition $Q$ satisfies $M_1$-reversibility. In this case, we
suggest to use $\rwQ{C}(x,\ud y) = \rwq{C}(x,y)\ud y$ with a
multivariate normal density $\rwq{C}(x,y) = N(y; x, C)$, with
covariance $C\in\R^{d\times d}$. In case of
constraints, that is, a non-trivial domain $D\subset\R^d$, we have $M_1 =
1(x\in D)$. Then, we suggest to use a Metropolis-Hastings type
transition probability:
\[
    \rwQ{C}(x,\ud y) = \rwq{C}(x,y)
\min\bigg\{ 1, \frac{M_1(y)}{M_1(x)} \bigg\} \ud y + \delta_x(\ud y)
r(x),
\]
where $r(x)\in[0,1]$ is the rejection probability.
This method works, of course, with arbitrary $M_1$, but our focus is
with a diffuse case, where the domain $D$ is regular and large enough, so that
rejections are rare. We stress that also in this case, $M_1(x) = 1(x\in
D)$ may be improper.
We refer to Algorithm \ref{alg:ai-cpf} with $\rwQ{C}$ as the
`fully diffuse initialisation' CPF (FDI-CPF).

We note that whenever $M_1$ can be evaluated pointwise, the FDI-CPF can always
be applied, by considering the modified Feynman-Kac model
$\tilde{M}_1\equiv 1$ and $\tilde{G}_1(x) = M_1(x) G_1(x)$.
However, when $M_1$ is Gaussian, the DGI-CPF can often lead
to a more efficient method. As with standard random-walk Metropolis
algorithms, choosing the covariance $C\in\R^{d\times d}$ is important
for the efficiency of the FDI-CPF.


\subsection{Adaptive proposals}
\label{sec:adapt} 

Finding a good autoregressive parameter of $\cnQ{\beta}$ or the
covariance parameter of $\rwQ{C}$ may be time-consuming in practice.
Inspired by the recent advances in adaptive MCMC
\citep[cf.][]{andrieu-thoms,vihola-amcmc}, it is natural to apply
adaptation also with the (iterated) AI-CPF. Algorithm \ref{alg:aai-cpf}
summarises a generic adaptive AI-CPF (AAI-CPF) using a parameterised family
$\{Q_\adaptstate\}_{\adaptstate\in\adaptspace}$ of $M_1$-reversible
proposals, with parameter $\adaptstate$.
\begin{algorithm}
    \caption{$\text{AAI-CPF}(\refr{x}_{1:T}^{(0)}; Q_{\adaptstate^{(0)}}, M_{2:T}, G_{1:T}, N)$}
     \label{alg:aai-cpf} 
\begin{algorithmic}[1]
\For{$j = 1, \ldots, n$}
\State $(\refr{x}_{1:T}^{(j)}, \adaptdata^{(j)}) \gets
          \text{AI-CPF}(\refr{x}_{1:T}^{(j-1)}; Q_{\adaptstate^{(j-1)}}, M_{2:T}, G_{1:T}, N)$.
\State $\adaptstate^{(j)} \gets \textsc{Adapt}(\adaptstate^{(j-1)}, \adaptdata^{(j)}, j)$.
\EndFor
\State \textbf{output}
$(\refr{x}_{1:T}^{(1)},\ldots,\refr{x}_{1:T}^{(n)})$.
\end{algorithmic}
\end{algorithm} 
The function $\textsc{Adapt}$ implements the adaptation, which
typically leads to $\adaptstate^{(j)} \to \adaptstate^*$,
corresponding to a well-mixing configuration.
We refer to the instances of the AAI-CPF with the AI-CPF step corresponding to the
DGI-CPF and the FDI-CPF as the adaptive DGI-CPF and FDI-CPF, respectively.

We next focus on concrete adaptations which may be used within our framework.
In the case of the FDI-CPF, Algorithm \ref{alg:adapt-fdi-am} implements a stochastic
approximation variant \cite{andrieu-moulines} of the adaptive
Metropolis covariance adaptation \cite{haario-saksman-tamminen}.
\begin{algorithm}
  \caption{\textsc{Adapt}\textsubscript{FDI, AM}$\big((\mu$, $\Sigma$),
    $(B_{1}, W_1^{(1:N)}, X_1^{(1:N)}$), $j\big)$}
      \label{alg:adapt-fdi-am} 
  \begin{algorithmic}[1]
    \State $\mu_* \gets (1 - \eta_j)\mu + \eta_j X_1^{(B_1)}$.
    \State $\Sigma_* \gets (1 - \eta_j)\Sigma + \eta_j
          (X_1^{(B_1)} - \mu)(X_1^{(B_1)} - \mu)^{\mathrm{T}}$.
    \State \textbf{output} $(\mu_*, \Sigma_*)$.
  \end{algorithmic}
\end{algorithm} 
Here, $\eta_j$ are step sizes that decay to zero, $\adaptstate_j =
(\mu_j,\Sigma_j)$ the estimated mean and covariance of the smoothing
distribution, respectively, and
$Q_\adaptstate = \rwQ{c \Sigma}$ where $c>0$ is a scaling factor of the
covariance $\Sigma$.
In the case of random-walk Metropolis, this scaling factor is usually
taken as $2.38^2/d$ \cite{gelman-roberts-gilks}, where $d$ is the state dimension of the model.
In the present context, however,
the optimal value of $c > 0$ appears to depend on the model and
on the number of particles $N$.
This adaptation mechanism can be used both with
\textsc{PickPath-AT} and with \textsc{PickPath-BS}, but may require some
manual tuning to find a suitable $c>0$.

Algorithm \ref{alg:adapt-fdi-aswam} details another adaptation for the FDI-CPF,
which is intended to be used together with \textsc{PickPath-BS} only. Here,
$\adaptstate_j = (\mu_j,\Sigma_j, \delta_j)$ contains the estimated mean,
covariance and the scaling factor, and
$Q_\adaptstate =
\rwQ{C(\adaptstate)}$, where $C(\adaptstate) =  e^\delta \Sigma$.

\begin{algorithm}
  \caption{\textsc{Adapt}\textsubscript{FDI, ASWAM}$\big((\mu$, $\Sigma$, $\delta$),
        ($B_{1}, W_{1}^{(1:N)}, X_{1}^{(1:N)}$), $j\big)$}
  \label{alg:adapt-fdi-aswam}
  \begin{algorithmic}[1]
    \State $\mu_* \gets (1 - \eta_j)\mu + \eta_j\sum_{i = 1}^{N}
    W_1^{(i)}X_1^{(i)}$.
    \State $\Sigma_* \gets (1 - \eta_j)\Sigma + \eta_j\sum_{i = 1}^{N}
    W_1^{(i)}
          (X_1^{(i)} - \mu)(X_1^{(i)} - \mu)^{\mathrm{T}}$.
    \State $\delta_* \gets \delta + \eta_j(\alpha - \alphatarget)$
    where $\alpha = 1 - W_1^{(1)}$.
    \State \textbf{output} $(\mu_*, \Sigma_*, \delta_*)$.
  \end{algorithmic}
\end{algorithm} 
%
\begin{algorithm}
  \caption{\textsc{Adapt}\textsubscript{DGI, AS}$\big(\adaptstate$,
        $(B_{1}, W_{1}^{(1:N)}, X_{1}^{(1:N)}$), $j\big)$}
      \label{alg:adapt-dgi} 
  \begin{algorithmic}[1]
    \State $\adaptstate_* \gets \adaptstate + \eta_j(\alpha - \alphatarget)$
    where $\alpha = 1 - W_1^{(1)}$.
    \State \textbf{output} $\adaptstate_*$.
  \end{algorithmic}
\end{algorithm} 
This algorithm is inspired by a Rao-Blackwellised variant of the
adaptive Metropolis within adaptive
scaling method \citep[cf.][]{andrieu-thoms}, which is applied with
standard random-walk Metropolis. We use
all particles with their backward sampling weights to update the
mean $\mu$ and covariance $\Sigma$, and an `acceptance rate' $\alpha$,
that is, the probability that the first coordinate of the reference
trajectory is not chosen.
Recall that after the AI--CPF in Algorithm \ref{alg:aai-cpf} has been run,
the first coordinate of the reference trajectory and its associated weight reside
in the first index of the particle and weight vectors contained in $\adaptdata^{(j)}$.

The optimal value of the acceptance rate
parameter $\alphatarget$ is
typically close to one, in contrast with random-walk
Metropolis, where $\alphatarget \in [0.234,0.44]$ are common
\citep{gelman-roberts-gilks}. Even though the optimal value appears to
be problem-dependent, we have found empirically that $0.7\le
\alphatarget \le 0.9$ often leads to reasonable mixing.
We will show empirical evidence for this finding in Section \ref{sec:experiments}.

Algorithm \ref{alg:adapt-dgi} describes a similar adaptive scaling type
mechanism for
tuning $\beta = \logit^{-1}(\adaptstate)$ in the DGI-CPF,
with $Q_\adaptstate = \cnQ{\beta}$.
The algorithm is most practical with \textsc{PickPath-BS}.

We conclude this section with a consistency result for
Algorithm \ref{alg:aai-cpf}, using the
adaptation mechanisms in Algorithms
\ref{alg:adapt-fdi-am} and \ref{alg:adapt-fdi-aswam}.
In Theorem \ref{thm:slln},
we denote $(\mu_j,\Sigma_j) = \adaptstate_j$
in the case of Algorithm \ref{alg:adapt-fdi-am},
and $(\mu_j,\Sigma_j,\delta_j) = \adaptstate_j$
with Algorithm \ref{alg:adapt-fdi-aswam}.

\begin{theorem}
    \label{thm:slln} 
Suppose $D$ is a compact set,
a uniform mixing condition (Assumption
\ref{a:mixing} in Appendix
\ref{app:adapt-proof}) holds,
and there exists an $\epsilon>0$ such that
for all $j\ge 1$,
the smallest eigenvalue $\lambda_{\min}(\Sigma_j)\ge \epsilon$,
and with Algorithm
\ref{alg:adapt-fdi-aswam} also $\delta_j\in [\epsilon,\epsilon^{-1}]$.
Then, for any bounded function $f:\X\to\infty$,
\[
    \frac{1}{n}\sum_{k=1}^n f(\refr{x}_{1:T}^{(k)})
    \xrightarrow{n\to\infty} \pi(f).\qquad \text{almost surely.}
\]
\end{theorem} 
The proof of Theorem \ref{thm:slln} is given in Appendix
\ref{app:adapt-proof}. The proof is slightly more general, and
accomodates for instance $t$-distributed instead of Gaussian proposals
for the FDI-CPF.
We note that the latter stability condition, that is, existence of
the constant $\epsilon>0$,
may be enforced by introducing a `rejection' mechanism in the
adaptation; see the end of Appendix \ref{app:adapt-proof}. However, we have found
empirically that the adaptation is stable also without such
a stabilisation mechanism.


\subsection{Use within particle Gibbs}
\label{sec:with-pg} 

Typical application of HMMs in statistics involves not only smoothing,
but also inference of a number of
`hyperparameters' $\theta$, with prior density $\mathrm{pr}(\theta)$,
and with
\begin{equation}
  \label{eq:hmm-with-params}
    p(y_{1:T},x_{1:T}\mid \theta)
    = \gamma_\theta(x_{1:T})
    = M_1(x_1) G_1^{(\theta)}(x_1) \prod_{k=2}^T
    M_k^{(\theta)}(x_{k-1},x_k) G_k^{(\theta)}(x_{k-1},x_k).
\end{equation}
The full posterior, $\check{\pi}(\theta, x_{1:T}) \propto \mathrm{pr}(\theta)
\gamma_\theta(x_{1:T})$ may be inferred with the particle Gibbs (PG) algorithm
\cite{pmcmc}. (We assume here that $M_1$ is diffuse, and thereby
independent of $\theta$.)

The PG alternates between (Metropolis-within-)Gibbs updates for
$\theta$ conditional on $x_{1:T}$, and CPF updates for $x_{1:T}$
conditional on $\theta$. The (A)AI-CPF applied with
$M_{2:T}^{(\theta)}$ and $G_{1:T}^{(\theta)}$ may be used as a
replacement of the CPF steps in a PG. Another adaptation, independent
of the AAI-CPF,  may be used for the hyperparameter updates; see for
instance the discussion in \cite{vihola-amcmc}.

Algorithm \ref{alg:aai-pg} summarises a generic adaptive PG with the AAI-CPF.
Line \ref{line:pg-param} involves an update of $\theta^{(j-1)}$ to
$\theta^{(j)}$ using transition probabilities
$K_{\adaptstate_\theta}(\uarg, \uarg \mid x_{1:T})$
which leave $\check{\pi}(\theta\mid x_{1:T})$ invariant,
and Line \ref{line:pg-param-adapt} is (optional) adaptation.
This could, for instance, correspond to the robust adaptive Metropolis
algorithm (RAM) as suggested in \cite{vihola-amcmc}.
Lines \ref{line:pg-ai-cpf} and \ref{line:pg-ai-cpf-adapt} implement
the AAI-CPF. Note that without Lines \ref{line:pg-param-adapt} and
\ref{line:pg-ai-cpf-adapt}, Algorithm \ref{alg:aai-pg}
determines a $\check{\pi}$-invariant transition rule.

\begin{algorithm}
    \caption{$\text{AAI-PG}(\theta^{(0)}, \refr{x}_{1:T}^{(0)};
      Q_{C(\adaptstate^{(0)})}, M_{2:T}, G_{1:T}, N)$}
     \label{alg:aai-pg} 
\begin{algorithmic}[1]
\For{$j = 1, \ldots, n$}
\State \label{line:pg-param} $(\theta^{(j)},\adaptdata_\theta^{(j)})
\sim K_{\adaptstate_\theta^{(j-1)}}(\theta^{(j-1)}, \uarg\mid
\refr{x}_{1:T}^{(j-1)})$.
\State \label{line:pg-param-adapt} $\adaptstate_\theta^{(j)} \gets
\textsc{Adapt}_\theta(\adaptstate_\theta^{(j-1)}, \theta^{(j)},
\adaptdata_\theta^{(j)})$.
\State \label{line:pg-ai-cpf} $(\refr{x}_{1:T}^{(j)}, \adaptdata^{(j)}) \gets
          \text{AI-CPF}(\refr{x}_{1:T}^{(j-1)};
          Q_{\adaptstate^{(j-1)}}, M_{2:T}^{(\theta^{(j)})},
          G_{1:T}^{(\theta^{(j)})}, N)$.
\State \label{line:pg-ai-cpf-adapt} $\adaptstate^{(j)} \gets \textsc{Adapt}(\adaptstate^{(j-1)},
\adaptdata^{(j)}, j)$.
\EndFor
\State \textbf{output}
$\big((\theta^{(1)},\refr{x}_{1:T}^{(1)}),\ldots,(\theta^{(n)},\refr{x}_{1:T}^{(n)})\big)$.
\end{algorithmic}
\end{algorithm}




\section{Experiments} \label{sec:experiments} 
In this section, we study the application of the methods presented in
Section \ref{sec:diffuse-init} in practice.
Our focus will be on the case of the bootstrap filter, that is, 
$M_1(\ud x_1) = p(x_1) \ud x_1$, $M_k(x_{k-1}, \ud x_k) = p(x_k\mid x_{k-1}) \ud x_k$ 
and $G_k(x_{1:k}) = p(y_k\mid x_k)$.

We start by investigating two simple HMMs: the noisy random walk
model (RW), that is,
\eqref{eq:noisyar} with $\rho = 1$,
and the following stochastic volatility (SV) model:
\begin{equation}
  \label{eq:sv}
  \begin{aligned}
    x_{k+1} &= x_{k} + \eta_k, \\
    y_{k} &= e^{x_k}\epsilon_k,
  \end{aligned}
\end{equation}
with $x_1 \sim N(0, \sigma_1^2)$, $\eta_k \sim N(0, \sigma_x^2)$ and
$\epsilon_k \sim N(0, \sigma_y^2)$.
In Section \ref{sec:mvnormal-experiments},
we study the dependence of the method with varying dimension,
with a static multivariate normal model.  We conclude in
Section \ref{sec:seir} by applying our methods in a realistic inference
problem related to modelling the COVID-19 epidemic in Finland.
%
\subsection{Comparing DGI-CPF and CPF-BS}
\label{sec:dgi-cpf-experiments} 

We first studied how the DGI-CPF performs in comparison to the CPF-BS when
the initial distributions of the RW and SV model are diffuse.
Since the efficiency of sampling is
affected by both the values of the model parameters (cf.~Figure
\ref{fig:diffinit-poor-mixing-example}) and the number of particles $N$, we
experimented with a range of values $N \in \{8, 16, 32, 64, 128, 256,
512\}$ for which we applied both methods with $n = 10000$ iterations
plus $500$ burn-in. We simulated data from both the
RW and SV models with $T = 50$, $x_1 = 0$, $\sigma_y = 1$ and
varying $\sigma_x \in \{0.01, 0.05, 0.1, 0.5, 1, 2, 5, 10, 20,
50, 100, 200\}$.
We then applied both methods for each dataset with the corresponding $\sigma_x$,
but with varying $\sigma_1 \in \{10, 50, 100, 200, 500, 1000\}$,
to study the sampling efficiency under different parameter
configurations ($\sigma_x$ and $\sigma_1$).
For the DGI-CPF, we varied
the parameter $\beta\in \{0.01, 0.02, \ldots, 0.99\}$.
We computed the estimated integrated autocorrelation time ($\iact$)
of the simulated values of $x_1$ and scaled this by the number of particles $N$.
The resulting quantity, the inverse relative efficiency ($\ire$), measures the asymptotic efficiencies of estimators with varying computational costs \citep{glynn}. 

Fig. \ref{fig:dgi-cpf-opt-vs-cpf-bs-NOISYAR-logiactn} shows the comparison of
the CPF-BS with the best DGI-CPF, that is, the DGI-CPF with the $\beta$ that
resulted in the lowest $\iact$ for each parameter configuration and $N$.

The results indicate that with N fixed, a successful tuning
of $\beta$ can result in greatly improved mixing in comparison with
the CPF-BS. While the performance of the CPF-BS approaches that of the best
DGI-CPF with increasing $N$, the difference in performance remains substantial with
parameter configurations that are challenging for the CPF-BS.

The optimal $N$ which minimizes the $\ire$ depends on the parameter
configuration. For `easy' configurations (where $\ire$ is small), even $N=8$
can be enough, but for more `difficult' configurations (where $\ire$ is large), higher 
values of $N$ can be optimal.
Similar results for the SV model are shown in the supplement Figure
\ref{fig:dgi-cpf-opt-vs-cpf-bs-SV-logiactn}, and lead to similar conclusions.
\begin{figure} 
  \centering
  \includegraphics[width=0.8\textwidth]{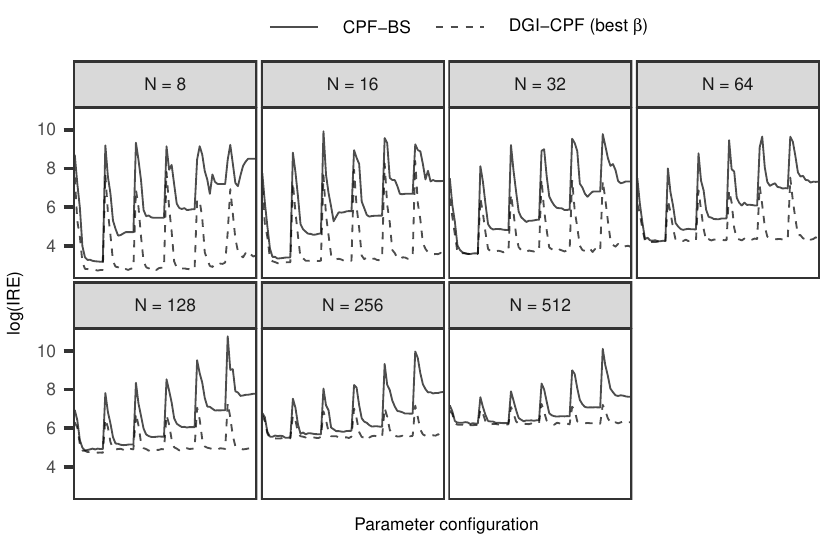}
  \caption{The $\log{(\ire)}$ resulting from the application of the CPF-BS and
  the best case DGI-CPF to the RW model.
  The horizontal axis depicts different configurations of
  $\sigma_1$ and $\sigma_x$, and in each panel $N$ varies.}
  \label{fig:dgi-cpf-opt-vs-cpf-bs-NOISYAR-logiactn}
\end{figure}

The varying `difficulty' of the parameter configurations is further illustrated in Figure
\ref{fig:cpf-bs-dgi-cpf-SV-log-iact-heatmap-npar256}, which shows the
$\log{(\iact)}$ for the SV model with $N = 256$ particles. 
The CPF-BS performed the worst when the initial distribution was
very diffuse with respect to the state noise $\sigma_x$, as expected. In contrast,
the well-tuned DGI-CPF appears rather robust with respect to changing parameter
configuration. The observations were similar with other $N$, and for the RW
model; see supplement Figure
\ref{fig:cpf-bs-dgi-cpf-NOISYAR-log-iact-heatmap-npar256}.
\begin{figure} 
  \centering
  \includegraphics[width=0.8\textwidth]{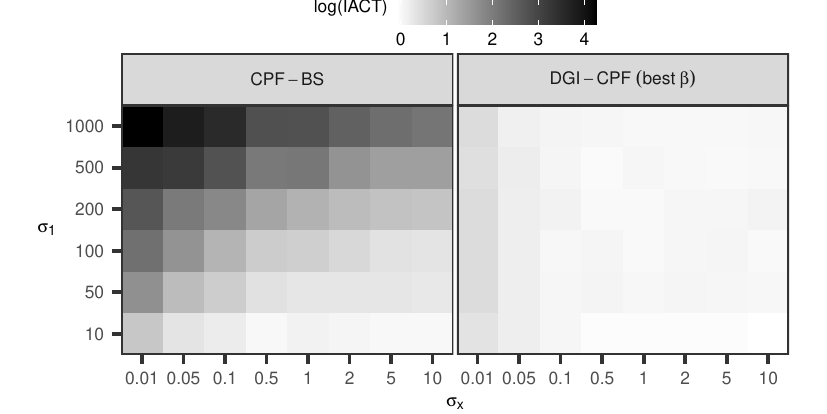}
  \caption{The $\log{(\iact)}$ of the CPF-BS (left) and the best case
  DGI-CPF (right) with respect to $\sigma_1$ and $\sigma_x$ in the case of the
  SV model and $N = 256$.}
  \label{fig:cpf-bs-dgi-cpf-SV-log-iact-heatmap-npar256}
\end{figure}

The results in Figures \ref{fig:dgi-cpf-opt-vs-cpf-bs-NOISYAR-logiactn} and
\ref{fig:cpf-bs-dgi-cpf-SV-log-iact-heatmap-npar256} illustrate the
potential of the DGI-CPF, but are overly optimistic because
in practice, the $\beta$ parameter of the DGI-CPF cannot be chosen optimally.
Indeed, the choice of $\beta$ can have a substantial effect on the mixing.
Figure \ref{fig:dgi-cpf-log-mean-iact-beta-reps-npar-128-sigmax-1_0-sigmax1-50-sv}
illustrates this in the case of the SV model
by showing the logarithm of the mean $\iact$ over replicate runs of the DGI-CPF, for a range of
$\beta$.
Here, a $\beta$ of approximately 0.125 seems to yield close to optimal performance,
but if the $\beta$ is chosen too low, the sampling efficiency is greatly reduced, rendering
the CPF-BS more effective.
\begin{figure}
  \centering
  \includegraphics[width=0.8\textwidth]{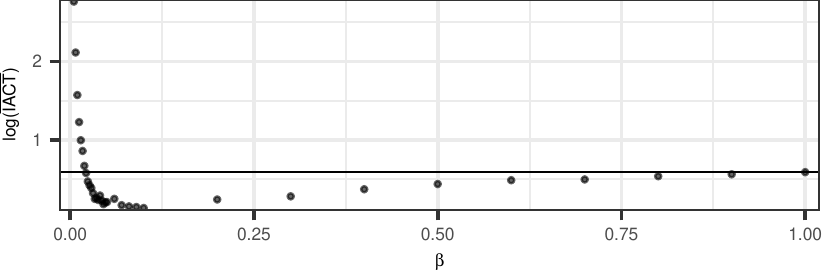}
  \caption{The logarithm of the mean $\iact$ over 5 replicate runs of the DGI-CPF with
  respect to varying $\beta$.
  The dataset was simulated from the SV model with parameters
  $\sigma_x = 1$ and $\sigma_{1} = 50$ and fixed in each replicate run of the
  algorithm. $N$ was set to $128$.
  The horizontal line depicts the performance of the CPF-BS.}
  \label{fig:dgi-cpf-log-mean-iact-beta-reps-npar-128-sigmax-1_0-sigmax1-50-sv}
\end{figure}

This highlights the importance of choosing an appropriate value for $\beta$, and
motivates our adaptive DGI-CPF, that is, Algorithm \ref{alg:aai-cpf} together with
Algorithm \ref{alg:adapt-dgi}.
We explored the effect of the target acceptance rate $\alphatarget
\in \{0.01, 0.02, \ldots, 1\}$, with the same
datasets and parameter configurations as before.
Figure \ref{fig:dgi-cpf-vs-cpf-bs-target-log-mean-iactn} summarises the
results for both the SV and RW models, in comparison with the CPF-BS.
The results indicate that with a wide range of target acceptance
rates, the adaptive DGI-CPF exhibits improved mixing over the CPF-BS.
When $N$ increases, the optimal values for $\alphatarget$ appear to tend
to one. However, in practice, we are interested in a moderate $N$, for
which the results suggest that the best candidates for values of
$\alphatarget$ might often be found in the range from $0.7$ to $0.9$.

For the CPF-BS, the mean $\ire$ is approximately constant, which might suggest
that the optimal number of particles is more than 512.
In contrast, for an appropriately tuned DGI-CPF, the mean $\ire$ is optimised 
by $N = 32$ in this experiment. 
\begin{figure}
  \centering
  \includegraphics[width=0.8\textwidth]{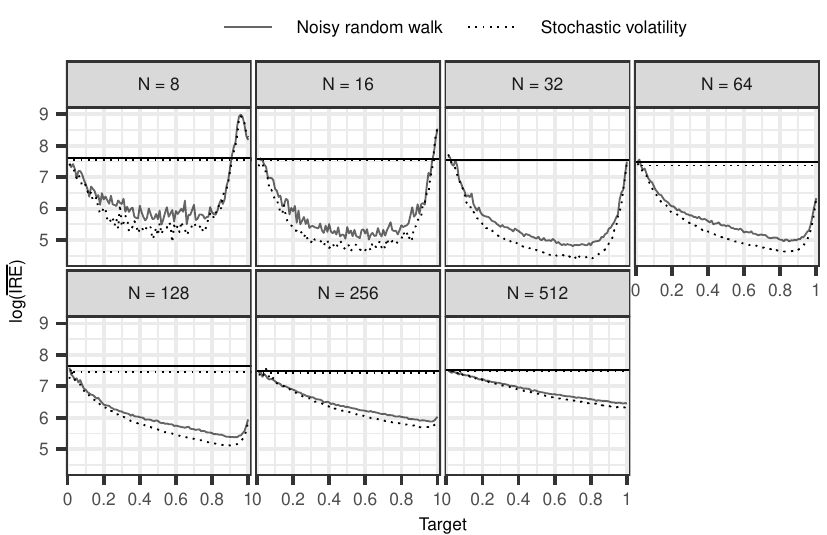}
  \caption{The logarithm of the mean $\ire$ over
  the parameter configurations with the adaptive DGI-CPF and
  varying target acceptance rates.
  The horizontal lines depict the performance of the CPF-BS.}
  \label{fig:dgi-cpf-vs-cpf-bs-target-log-mean-iactn}
\end{figure}
%

\subsection{Comparing FDI-CPF and particle Gibbs}
\label{sec:fdi-cpf-experiments} 
Next, we turn to study a fully diffuse
initialisation. In this case, $M_1$ is improper, and we cannot use the CPF directly.
Instead, we compare the performance of the adaptive FDI-CPF
with what we call the diffuse particle Gibbs (DPG-BS) algorithm.
The DPG-BS is a standard particle Gibbs algorithm, where the first
latent state $x_1$ is
regarded as a `parameter', that is, the algorithm alternates between
the update of $x_1$ conditional
on $x_{2:T}$ using a random-walk Metropolis-within-Gibbs step,
and the update of the latent state variables
$x_{2:T}$ conditional on $x_1$ using the CPF-BS.
We also adapt the Metropolis-within-Gibbs proposal distribution
$Q_{\mathrm{DPG}}$ of the DPG-BS, using the RAM algorithm \cite{ram};
see also the discussion in \cite{vihola-amcmc}.
For further details regarding our implementation of the DPG-BS,
see Appendix \ref{app:dpg_details}.

We used a similar simulation experiment as with the adaptive
DGI-CPF in Section \ref{sec:dgi-cpf-experiments},
but excluding $\sigma_1$, since the initial distribution was now fully diffuse.
The target acceptance rates in the FDI-CPF with the ASWAM adaptation were again varied
in $\alphatarget \in \{0.01, 0.02, \ldots, 1\}$ and the scaling factor in the AM
adaptation was set to $c = 2.38^2$.
In the DPG-BS, the target acceptance rate for updates of the initial state using
the RAM algorithm was fixed to 0.441 following
\cite{gelman-roberts-gilks}.

Figure \ref{fig:NOISYAR-fdi-cpf-vs-dpg-cpf-logiactn} shows results with the RW
model for the DPG-BS, the FDI-CPF with the AM adaptation, and the FDI-CPF with
the ASWAM adaptation using the best value for $\alphatarget$. The FDI-CPF
variants appear to perform better and improve upon the performance of
the DPG-BS especially with small $\sigma_x$. Similar to Fig. \ref{fig:dgi-cpf-opt-vs-cpf-bs-NOISYAR-logiactn} and
\ref{fig:cpf-bs-dgi-cpf-SV-log-iact-heatmap-npar256}, the optimal
$N$ minimizing the $\ire$ depends on the value of
$\sigma_x$: smaller values of $\sigma_x$ call for higher
number of particles. 

The performance of the adaptive FDI-CPF appears similar regardless of the
adaptation used, because the chosen scaling factor $c =
2.38^2$ for a univariate model was close to the optimal
value found by the ASWAM variant in this example. We experimented
also with $c = 1$, which led to less efficient AM, in
the middle ground between the ASWAM and the DPG-BS.

The $\iact$ for the DPG-BS stays approximately constant with increasing $N$,
which results in a $\log{(\ire)}$ that increases roughly by a constant
as $N$ increases.  This is understandable, because
in the limit as $N\to\infty$, the CPF-BS (within the DPG-BS) will correspond to a Gibbs
step, that is, a perfect sample of $x_{2:T}$ conditional on $x_1$. Because of the
strong correlation between $x_1$ and $x_2$, even an `ideal' Gibbs
sampler remains inefficient, and the small variation seen in the panels for the DPG-BS
is due to sampling variability.
The results for the SV model,
with similar findings, are shown in the supplement Figure
\ref{fig:SV-fdi-cpf-vs-dpg-cpf-logiactn}.
\begin{figure}
  \centering
  \includegraphics[width=0.8\textwidth]{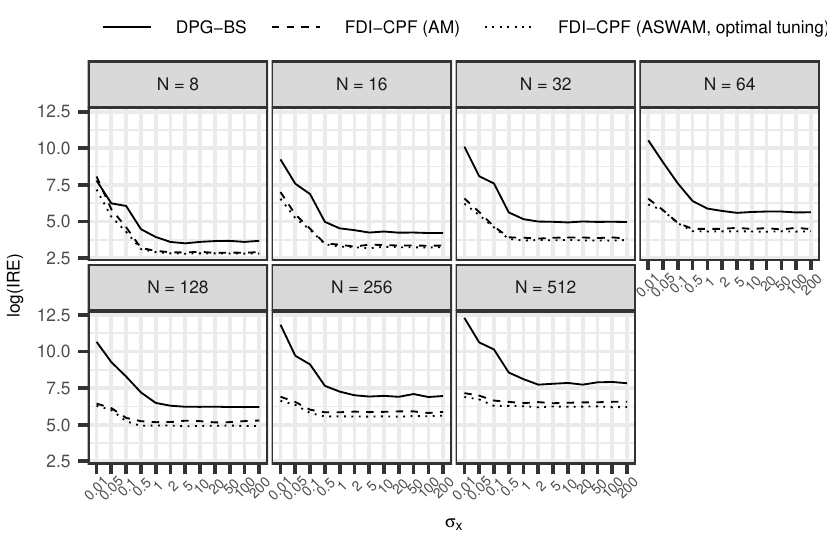}
  \caption{The $\log{(\ire)}$ for the DPG-BS, FDI-CPF with the AM adaptation and the best case
  FDI-CPF with the ASWAM adaptation to the datasets generated with varying $\sigma_x$ from
  the RW model.}
  \label{fig:NOISYAR-fdi-cpf-vs-dpg-cpf-logiactn}
\end{figure}

Figure \ref{fig:fdi-dpg-log-mean-iactn-wrt-tuning} shows the logarithm of the mean
$\ire$ of the FDI-CPF with the ASWAM adaptation with respect to
varying target acceptance rate $\alphatarget$. The results are reminiscent
of Figure \ref{fig:dgi-cpf-vs-cpf-bs-target-log-mean-iactn} and show
that with a moderate fixed $N$, the FDI-CPF with the ASWAM adaptation
outperforms the DPG-BS with a wide range of values for $\alphatarget$. The
optimal value of $\alphatarget$ seems to tend to one as $N$ increases, but
again, we are mostly concerned with moderate $N$.
For a well-tuned FDI-CPF the minimum mean $\ire$ is found when $N$ is roughly 
between 32 and 64. 
\begin{figure}
  \centering
  \includegraphics[width=0.8\textwidth]{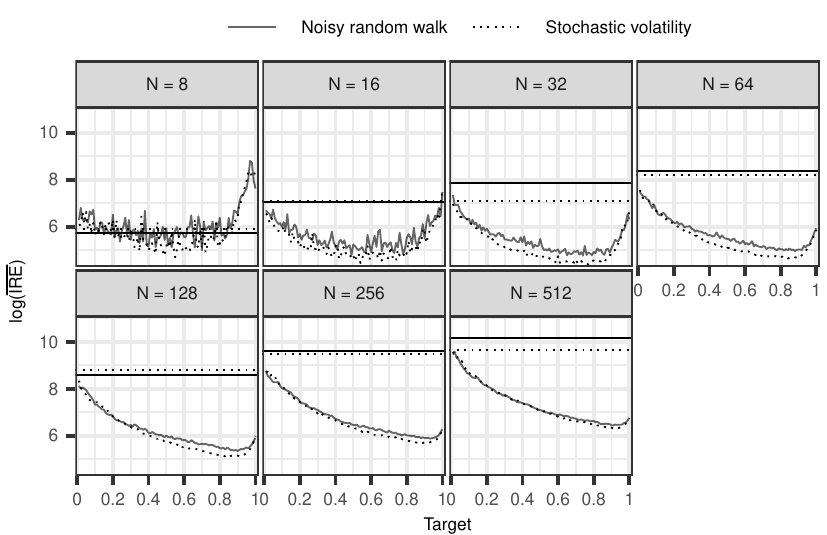}
  \caption{A comparison of the FDI-CPF with the ASWAM adaptation against the DPG-BS.
  The horizontal axis shows the target acceptance rate $\alphatarget$
  used in the adaptive FDI-CPF.
  The logarithm of the mean $\ire$ on the vertical axis is computed over
  the different $\sigma_x$ values.
  The black horizontal lines show the performance with the DPG-BS.}
  \label{fig:fdi-dpg-log-mean-iactn-wrt-tuning}
\end{figure}

\subsection{The relationship between state dimension, number of particles and
optimal target acceptance rate} \label{sec:mvnormal-experiments} 
A well chosen value for the target acceptance rate $\alphatarget$ appears
to be key for obtaining good performance with the adaptive DGI-CPF and
the FDI-CPF with the ASWAM adaptation. In Sections
\ref{sec:dgi-cpf-experiments}--\ref{sec:fdi-cpf-experiments}, we
observed a relationship between $N$ and the optimal target acceptance
rate, denoted here by $\alphaopt$, with two univariate HMMs. It is
expected that $\alphaopt$ is generally somewhat model-dependent, but
in particular, we suspected that the methods might behave differently
with models of different state dimension $d$.

In order to study the relationship between $N$, $d$ and $\alphaopt$
in more detail, we considered a simple multivariate normal model with
$T = 1$, $M_1(x) \propto 1$, and $G_1(x_1) = N(x_1; 0, \sigma I_d)$, the
density of $d$ independent normals.
We conducted a simulation experiment with 6000 iterations
plus 500 burn-in. We applied the FDI-CPF with the ASWAM adaptation with
all combinations of $N \in \{2^4, 2^5, \ldots, 2^{11}\}$,
$\alphatarget \in \{0.01, 0.02, \ldots, 1\}$,
$\sigma \in \{1, 5, 10, 50, 100\}$,
and with dimension $d \in \{1, 2, \ldots, 10\}$.
Unlike before, we monitor the $\iact$ over the samples of $x_1$ as an 
efficiency measure.

Figure \ref{fig:fdi-cpf-aswam-mvnormal-target-vs-log-mean-iact} summarises
the results of this experiment.
With a fixed state dimension, $\alphaopt$ tended towards 1 with
increasing numbers of particles $N$, as observed with the RW and SV models
above.
With a fixed number of particles $N$, $\alphaopt$ appears to
get smaller with increasing state dimension $d$, 
but the change rate appears slower with higher $d$.
Again, with moderate values for $N$ and $d$, the values in the range 0.7--0.9 seem
 to yield good performance.
\begin{figure}
  \centering
  \includegraphics[width=0.7\textwidth]{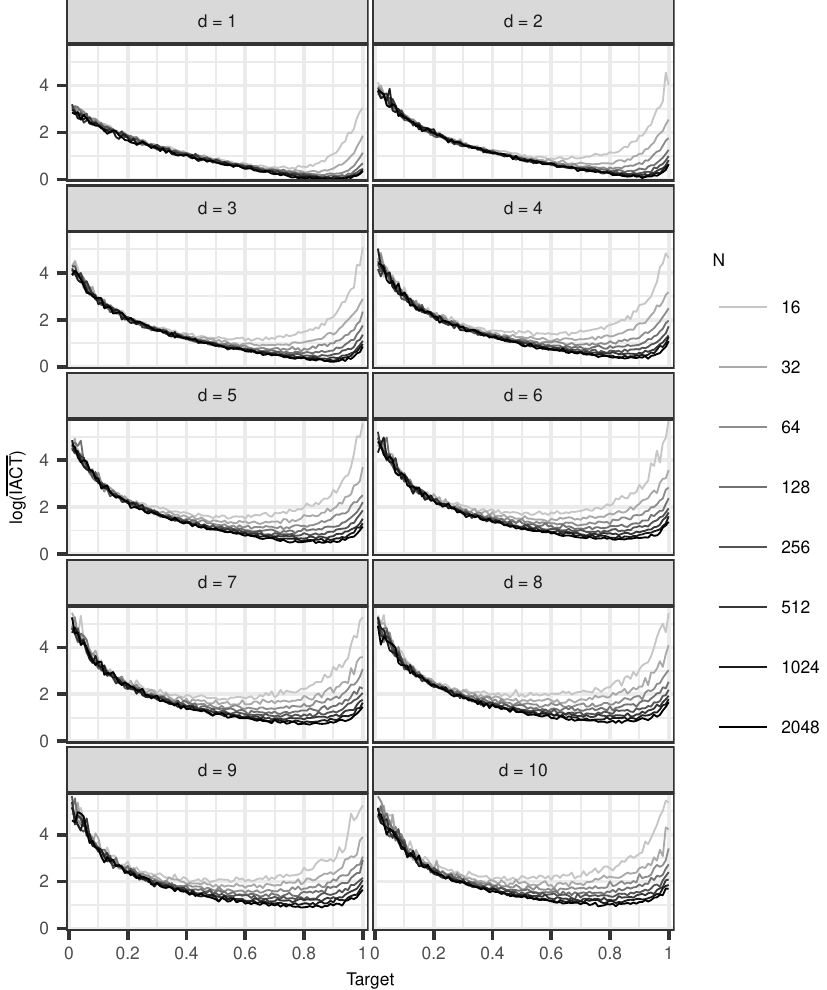}
  \caption{The effect of state dimension $d$, number of particles $N$ and
    target acceptance rate $\alphatarget$
           on the logarithm of the mean $\iact$ in the multivariate normal model.
           The means are computed over the different $\sigma$
           in the simulation experiment.}
  \label{fig:fdi-cpf-aswam-mvnormal-target-vs-log-mean-iact}
\end{figure}

Figure \ref{fig:alphaopt-vs-log-npar-dim} shows a different view of the same
data:
$\logit{(\alphaopt)}$ is plotted with respect to $\log{(N)}$ and $d$.
Here, we computed $\alphaopt$ by taking the target acceptance rate
that produced the lowest $\iact$ in the simulation experiment, for each value
of $\sigma$, $N$ and $d$.
At least with moderate $\alphaopt$ and $N$,
there appears to be a roughly linear relationship between
$\logit(\alphaopt)$ and $\log(N)$,
when $d$ is fixed.
However, because of the lack of theoretical backing, we do not suggest
to use such a simple model for choosing $\alphaopt$ in practice.
\begin{figure}
  \centering
  \includegraphics[width=0.8\textwidth]{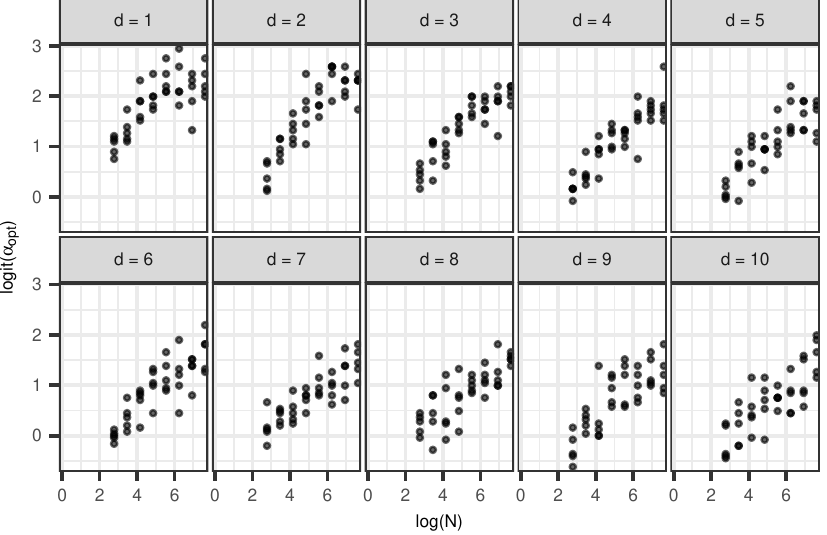}
  \caption{The best target acceptance rate $\alphaopt$ with respect to the
    the number of particles $N$ and state dimension $d$ on
    the multivariate normal model.}
  \label{fig:alphaopt-vs-log-npar-dim}
\end{figure}
%

\subsection{Modelling the COVID-19 epidemic in Finland}
\label{sec:seir} 
Our final experiment is a realistic inference problem arising from the
modelling of the progress of the COVID-19 epidemic in Uusimaa, the
capital region of Finland. Our main interest is in estimating the
time-varying transmission rate, or the basic reproduction number
$\rnaught$, which is expected to change over time, because of a number
of mitigation actions and social distancing. The
model consists of a discrete-time `SEIR' stochastic compartment model,
and a dynamic model for $\rnaught$; such epidemic models have been used earlier
in different contexts \citep[e.g.][]{shubin-lebedev-lyytikainen-auranen}.

We use a simple SEIR without age/regional stratification. That is,
we divide the whole population $\popsize$ to
four separate states:
susceptible ($S$), exposed ($E$), infected ($I$) and removed ($R$), so
that $\popsize = S + E + I + R$, and assume that
$\popsize$ is constant.
We model the transformed $\rnaught$, denoted by $\rho$, such that
$\rnaught = \rnaught^{\mathrm{max}}\logit^{-1}(\rho)$, where
$\rnaught^{\mathrm{max}}$ is the maximal value for $\rnaught$.
The state vector of the model at time $k$ is, therefore,
$X_k = (S_k, E_k, I_k, R_k, \rho_k)$.
The transition probability of the model is:
\begin{equation*}
  \begin{aligned}
    S_{k + 1} &= S_k - \Delta E_{k + 1},
    &\Delta E_{k+1} \sim \Binom(S_k, p_\beta),
    \ p_\beta &= 1 - \exp{(- \beta_k (I_k / \popsize))}, \\
    E_{k + 1} &= E_k + \Delta E_{k + 1} - \Delta I_{k + 1},
    &\Delta I_{k+1} \sim \Binom(E_k, p_a),
    \ p_a &= 1 - \exp{(- a)}, \\
    I_{k + 1} &= I_k + \Delta I_{k + 1} - \Delta R_{k + 1},
    &\Delta R_{k + 1} \sim \Binom(I_k, p_\gamma),
    \ p_\gamma &= 1 - \exp{(- \gamma)}, \\
    R_{k + 1} &= R_k + \Delta R_{k + 1}, & &\\
  \end{aligned}
\end{equation*}
and $\rho_{k+1} = \rho_k + \Delta \rho_{k + 1},$ with
$\Delta \rho_{k + 1} \sim \Normal(0,  \sigma^2)$.
Here, $\beta_k = \rnaught^{\mathrm{max}}\logit^{-1}(\rho_k)
p_\gamma$ is the time-varying infection rate,
and $a^{-1}$ and $\gamma^{-1}$ are the mean incubation period and recovery time, respectively.
Finally, the random walk parameter $\sigma$ controls how fast
$(\rho_k)_{k \geq 2}$ can change.

The data we use in the modelling consist of the daily number of individuals tested
positive for COVID-19 in Uusimaa \cite{thl-covid-data}.
We model the counts with a negative binomial distribution dependent on the number
of infected individuals:
\begin{equation}
  \label{eq:seir_obs}
  Y_k \sim \NegBinom\left(e p_{\gamma}\dfrac{p}{1 - p}I_k, p\right).
\end{equation}
Here, the parameter $e$ denotes sampling effort, that is, the average
proportion of infected individuals that are observed, and $p$ is the failure
probability of the negative binomial distribution, which controls the variability
of the distribution.

In the beginning of the epidemic, there is little information available regarding
the initial states, rendering the diffuse initialisation a
convenient strategy.
We set
\begin{equation}
  \label{eq:seir-M1}
  M_1(S_1, E_1, I_1, R_1, \rho_1) = 1(S_1 + E_1 + I_1 = \popsize)
     1(S_1, E_1, I_1 \geq 0)1(R_1 = 0),
\end{equation}
where the number of removed $R_1 = 0$ is justified because we assume
all were susceptible to COVID-19, and that the epidemic has started
very recently.

In addition to the state estimation, we are interested in estimating the
parameters $\sigma$ and $p$.
We assign the prior $N(-2.0, (0.3)^2)$
to $\log{(\sigma)}$ to promote gradual changes in $\rnaught$, and an
uninformative prior, $N(0, 10^2)$, for
$\logit(p)$.
The remaining parameters are fixed to
$\popsize = 1638469$, $\rnaught^{\mathrm{max}} = 10$, $a = 1/3$, $\gamma = 1/7$
and $e = 0.15$, which are in part inspired by the values
reported by the Finnish Institute for Health and Welfare.

We used the AAI-PG (Algorithm \ref{alg:aai-pg}) with the FDI-CPF with the ASWAM
adaptation, and a RAM adaptation
\citep{ram, vihola-amcmc} for $\sigma$ and $p$, (i.e.~in the lines
\ref{line:pg-param}--\ref{line:pg-param-adapt} of Algorithm
\ref{alg:aai-pg}). The form of \eqref{eq:seir-M1} leads to the version
of the FDI-CPF discussed in Section \ref{sec:fdi-cpf} where the
initial distribution is uniform with constraints.
We use a random-walk proposal to generate proposals
  $(\rho_1,E_1,I_1)\to (\rho_1^*,E_1^*,I_1^*)$, round $E_1^*$ and
$I_1^*$ to the nearest integer, and then set
$R_1^* = 0$ and $S_1^* = \popsize - E_1^* - I_1^* - R_1^{*}$.
We refer to this
variant of the AAI-PG as the FDI-PG algorithm.
Motivated by our findings in Sections \ref{sec:dgi-cpf-experiments}--\ref{sec:mvnormal-experiments}, we set the target
acceptance rate $\alphatarget$ in the FDI-CPF (within the FDI-PG) to 0.8.

As an alternative to the FDI-PG we also used a particle Gibbs algorithm
that treats
$\sigma$, $p$ as well as the initial states $E_1$, $I_1$ and $\rho_1$ as
parameters, using the RAM to adapt the random-walk proposal
\citep{ram,vihola-amcmc}.
This algorithm is the DPG-BS detailed in Appendix \ref{app:dpg_details}
with the difference that the parameters $\sigma$ and $p$ are updated
together with the initial state, and $p^{\mathrm{DPG}}$ additionally contains all terms of \eqref{eq:hmm-with-params} which depend on $\sigma$ and $p$.

We ran both the FDI-PG and the DPG-BS with $N = 64$ a total of $n=500,000$
iterations plus $10,000$ burn-in, and thinning of 10. Figures
\ref{fig:seir-ac-per-method-and-param} and \ref{fig:seir-traces}
show the first 50 autocorrelations and traceplots of
$E_1$, $I_1$, $(\rnaught)_1$, $\sigma$ and $p$, for both methods,
respectively. The corresponding $\iact$ and $\neff$ as well as
credible intervals for the means of these variables are shown in Table
\ref{tab:mix-stats-fdi-dpg}. The FDI-PG outperformed the DPG-BS with
each variable. However, as is seen from the supplement Figure
\ref{fig:seir-rel-iact-fdi-dpg}, the difference is most notable with
the initial states, and the relative performance of the DPG-BS approaches
that of the FDI-PG with increasing state index. The slow improvement in
the mixing of the state variable $R$ occurs because of the cumulative
nature of the variable in the model, and the slow mixing of early
values of $I$. We note that even though the mixing with the DPG-BS was
worse, the inference with $500,000$ iterations leads in practice to
similar findings. However, the FDI-PG could provide reliable inference
with much less iterations than the DPG-BS. The marginal density estimates of
the initial states and parameters are shown in the supplement Figure
\ref{fig:seir-dens-per-method-and-param}. The slight discrepancies in
the density estimates of $E_1$ and $I_1$ between the methods are
likely because of the poor mixing of these variables with the DPG-BS.

\begin{figure}
  \centering
  \includegraphics[width=0.8\textwidth]{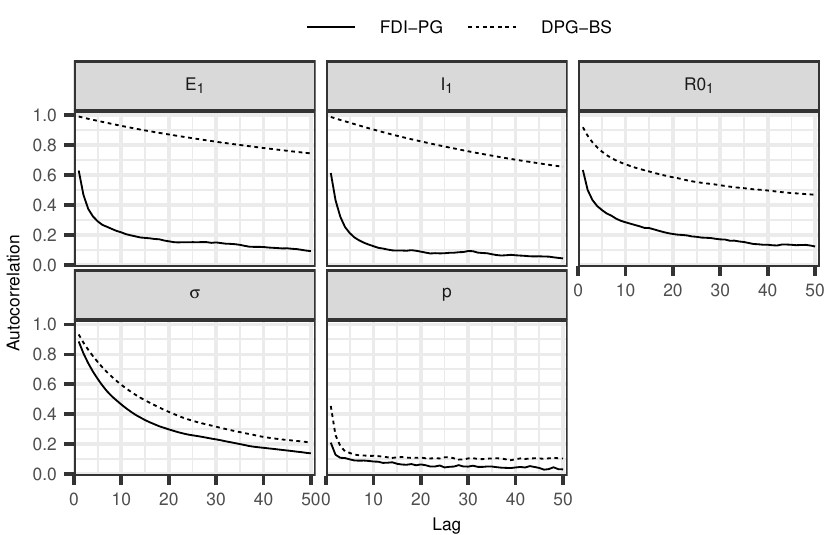}
  \caption{The first 50 autocorrelations for the model parameters and initial
  states with the FDI-PG and the DPG-BS computed after thinning the total 500000
  samples to every 10th sample.}
  \label{fig:seir-ac-per-method-and-param}
\end{figure}
\begin{figure}
  \centering
  \includegraphics[width=0.8\textwidth]{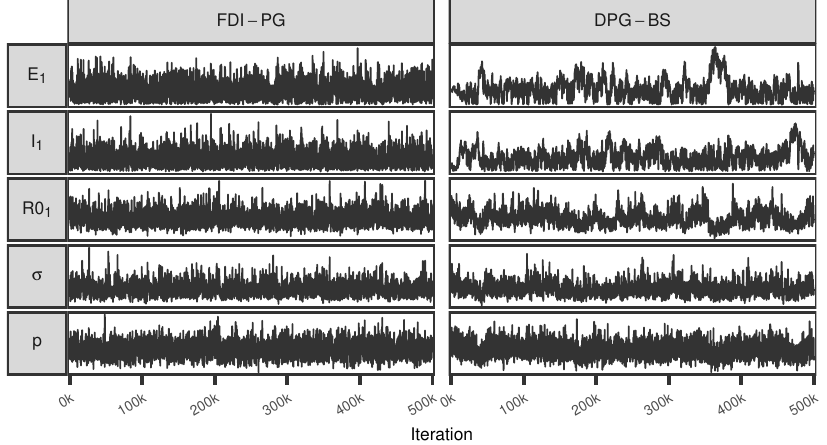}
  \caption{Traceplots for the initial states and model parameters
  for the SEIR model with the FDI-PG and the DPG-BS.
  The 5000 samples shown per method and parameter correspond to every 100th
  sample of the total 500000 samples simulated.}
  \label{fig:seir-traces}
\end{figure}
\begin{table}[h!]
\centering
\caption{The integrated autocorrelation time, effective sample size and credible
intervals of the mean for the initial states and parameters in the SEIR
model.}
\label{tab:mix-stats-fdi-dpg}
\ra{1.0}
\begin{tabular}{@{}ccccccc@{}}
\toprule
& \multicolumn{2}{c}{$\iact$} & \multicolumn{2}{c}{$\neff$} &
  \multicolumn{2}{c}{95\% mean CI} \\
\cmidrule(lr){2-3}\cmidrule(lr){4-5}\cmidrule(lr){6-7}
Variable & FDI-PG & DPG-BS & FDI-PG & DPG-BS & FDI-PG & DPG-BS \\
\midrule
$E_1$ & 30.087 & 882.583 & 1661.838 & 56.652 & (353.888, 374.054) & (301.379, 423.106)\\
$I_1$ & 14.296 & 626.963 & 3497.603 & 79.75 & (165.697, 172.388) & (155.374, 203.458)\\
$(\rnaught)_1$ & 32.168 & 436.755 & 1554.331 & 114.481 & (3.41, 3.513) & (3.266, 3.636)\\
$\sigma$ & 41.261 & 114.919 & 1211.796 & 435.088 & (0.15, 0.154) & (0.147, 0.154)\\
$p$ & 5.18 & 38.178 & 9652.794 & 1309.647 & (0.134, 0.135) & (0.133, 0.135)
\\\bottomrule
\end{tabular}
\end{table}


We conclude with a few words about our findings regarding the
changing transmission rate, which may be of some independent interest.
Figure \ref{fig:seir-fd-r0-and-postpred} displays the data and a
posterior predictive simulation, and the estimated distribution of $\rnaught$
computed by the FDI-PG with respect to time, with annotations about
events that may have had an effect on the spread of the epidemic, and/or the data.
The initial $\rnaught$ is likely somewhat overestimated, because of
the influx of infections from abroad, which were not explicitly modelled.
There is an overall decreasing trend since the beginning of `lockdown',
that is, when the government introduced the first mitigation
actions, including school closures. Changes in the testing criteria likely cause some bias soon
after the change, but no single action or event stands out.

Interestingly, if we look at our analysis, but restrict our focus
up to the end of April, we might be tempted to quantify how much certain
mitigation actions contribute to the suppression of the transmission
rate in order to build projections using scenario models
\cite[cf.][]{anderson-etal}.
However, when the mitigation measures have been gradually lifted
by opening the schools and restaurants, the openings do not appear to
have had notable consequences, at least until now. It is possible that
at this point, the number of infections was already so low, that it
has been possible to test all suspected cases and trace contacts so
efficiently, that nearly all transmission chains have been contained.
Also, the public may have changed their behaviour, and are now
following the hygiene and social distancing recommendations
voluntarily. Such a behaviour is, however, subject to change over time.

\begin{figure}
  \centering
  \includegraphics[width=\textwidth]{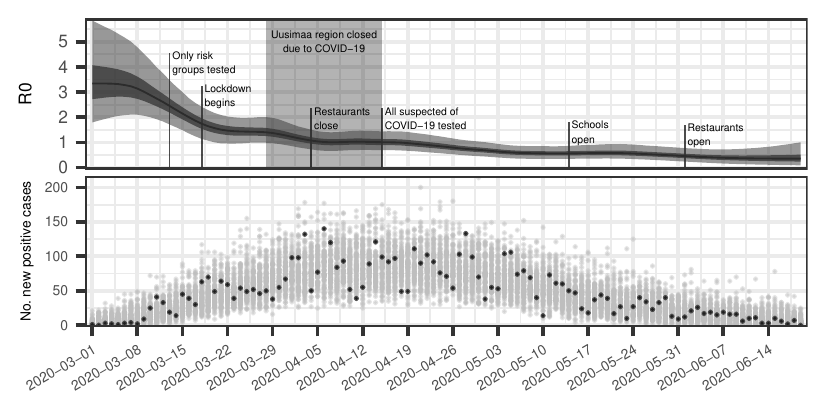}
  \caption{The distribution of the basic reproduction number $\rnaught$ (top) and
           a posterior predictive simulation (bottom) based on the posterior
           distribution computed with the FDI-PG.
           The plot for $\rnaught$ shows the median in black and probability
           intervals (75\% and 95\%) in shades of gray.
           The black points in the bottom plot represent the data used.
           The gray points represent observations simulated conditional on the posterior
           distribution of the model parameters and states.}
    \label{fig:seir-fd-r0-and-postpred}
\end{figure}


\section{Discussion} \label{sec:discussion} 

We presented a simple general auxiliary variable method for the CPF
for HMMs with diffuse initial distributions, and focused on two
concrete instances of it: the FDI-CPF for a uniform
initial density $M_1$ and the DGI-CPF for a Gaussian $M_1$.  We
introduced two mechanisms to adapt the FDI-CPF automatically:
the adaptive Metropolis (AM) \cite{haario-saksman-tamminen} and a method
similar to a Rao-Blackwellised adaptive scaling within adaptive Metropolis
(ASWAM) \citep[cf.][]{andrieu-thoms}, and provided a proof of their
consistency. We also suggested an adaptation for the DGI-CPF, based on
an acceptance rate optimisation. The FDI-CPF or the DGI-CPF, including
their adaptive variants, may be used directly within a particle Gibbs
as a replacement for the standard CPF.

Our experiments with a noisy random walk model and a stochastic
volatility model demonstrated that the DGI-CPF and the FDI-CPF can
provide orders of magnitude speed-ups relative to a direct application
of the CPF and to diffuse initialisation using particle Gibbs,
respectively. Improvement was substantial also in our motivating
practical example, where we applied the adaptive FDI-CPF (within
particle Gibbs) in the analysis of the COVID-19 epidemic in Finland,
using a stochastic `SEIR' compartment model with changing transmission
rate. Latent compartment models are, more generally, a good example
where our approach can be useful: there is substantial uncertainty
in the initial states, and it is difficult to design directly
a modified model that leads to efficient inference.

Our adaptation schemes are based on the estimated covariance matrix,
and a scaling factor which can be adapted using acceptance rate optimisation.
For the latter, we found empirically that with a moderate number of
particles, good performance was often reached with a target acceptance
rate ranging in 0.7--0.9. We emphasise that even though we found this `0.8 rule'
to work well in practice, it is only a heuristic, and the optimal
target acceptance rate may depend on the model of interest.  Related
to this, we investigated how the optimal target acceptance rate varied
as a function of the number of particles and state dimension in a
multivariate normal model, but did not find a clear pattern.
Theoretical verification of the acceptance rate heuristic, and/or
development of more refined adaptation rules, are left for future
research. We note that while the AM adaptation performed well in our
limited experiments, the ASWAM may be more appropriate when used
within particle Gibbs; see the discussion in \citep{vihola-amcmc}. The
scaling of the AM remains similarly challenging, due to the lack of theory
for tuning.


\section*{Acknowledgements} 

This work was supported by Academy of Finland grant 315619. We wish to acknowledge
CSC, IT Center for Science, Finland, for computational
resources, and thank Arto Luoma for inspiring discussions that led to the
COVID-19 example.


\bibliographystyle{abbrvnat}
\bibliography{refs}

\appendix

\section{Proof of Theorem \ref{thm:slln}}
\label{app:adapt-proof} 

For a finite signed measure $\xi$,
the total variation of $\xi$ is defined as
$\|\xi\|_{\mathrm{tv}} = \sup_{\|f\|_\infty\le 1} \xi(f)$,
where $\|f\|_\infty = \sup_x |f(x)|$, and the supremum is over
measurable real-valued functions $f$, and $\xi(f) = \int f \ud \xi$.
For Markov transitions $P$ and $P'$, define $d(P,P') = \sup_x \|
P(x,\uarg) - P'(x,\uarg) \|_{\mathrm{tv}}$.

In what follows, we adopt the following definitions:
\begin{definition} 
Consider Lines \ref{dcpf:f-cpf} and \ref{dcpf:cpf-select}
of Algorithm \ref{alg:ai-cpf} with
$\tilde{X}_1^{(1:N)}=\tilde{x}_1^{(1:N)}$ and $\refr{x}_{2:T}$, and define:
\begin{enumerate}[(i)]
    \item
      \label{item:cpf-path-space}
      $P_{\mathrm{CPF}}(\tilde{x}_1^{(1:N)},\refr{x}_{2:T}; \uarg)$
      as the law of $\tilde{X}_{1:T}^{(B_{1:T})}$, and
    \item
      \label{item:cpf-augmented-space}
      (In case \textsc{PickPath-BS} is used:)
      $\tilde{P}_{\mathrm{CPF}}(\tilde{x}_1^{(1:N)},\refr{x}_{2:T};
      \uarg)$
      as the law of $\big(\tilde{X}_{1:T}^{(B_{1:T})},
      (B_{1},V^{(1:N)}, \tilde{X}_{1}^{(1:N)})\big)$.
\end{enumerate}
Consider then Algorithm \ref{alg:ai-cpf} with parameterised
$Q=Q_\adaptstate$, and define, analogously:
\begin{enumerate}[resume*]
    \item
      \label{item:path-space}
      $P_\adaptstate$ is the Markov transition from
      $\refr{x}_{1:T}$ to $\tilde{X}_{1:T}^{(B_{1:T})}$.
    \item
      \label{item:augmented-space}
      $\tilde{P}_\adaptstate$
      is the Markov transition from
      from $(\refr{x}_{1:T},\uarg)$ to
      $\big(\tilde{X}_{1:T}^{(B_{1:T})},
      (B_{1},V^{(1:N)}, \tilde{X}_{1}^{(1:N)})\big)$.
\end{enumerate}
\end{definition}

\begin{lemma}
    \label{lem:kernel-diff} 
We have $d(P_\adaptstate, P_{\adaptstate'})
   \le N d(Q_\adaptstate, Q_{\adaptstate'})$
   and $d(\tilde{P}_\adaptstate, \tilde{P}_{\adaptstate'})
   \le N d(Q_\adaptstate, Q_{\adaptstate'})$.
\end{lemma}
\begin{proof} 
Let $(\hat{P}_\mathrm{CPF}, \hat{P}_\adaptstate)
\in \{(P_\mathrm{CPF},
  P_\adaptstate),(\tilde{P}_\mathrm{CPF},\tilde{P}_\adaptstate)\}$
and take measurable real-valued function $f$ on the state space of
$\hat{P}_\adaptstate$ with $\|f\|_\infty=1$.

We may write
\begin{equation}
    \hat{P}_\adaptstate(\refr{x}_{1:T}, f)
    = \int Q_\adaptstate(\refr{x}_1,\ud x_0)
    \bigg[\int \delta_{\refr{x}_1}(\ud
    \tilde{x}_1^{(1)})\prod_{k=2}^N
    Q_\adaptstate(x_0,\ud \tilde{x}_1^{(k)})
    \hat{P}_{\mathrm{CPF}}(\tilde{x}_{1}^{(1:N)}, \refr{x}_{2:T},
    f)\bigg],
    \label{eq:ai-cpf-decomposition}
\end{equation}
and therefore, upper bound
\begin{align*}
    &| \hat{P}_\adaptstate(\refr{x}_{1:T}, f) -
    \hat{P}_{\adaptstate'}(\refr{x}_{1:T}, f) | \\
    &\le |Q_\adaptstate(\refr{x}_1,g_0^{(\refr{x}_{1:T})}) -
    Q_{\adaptstate'}(\refr{x}_1,g_0^{(\refr{x}_{1:T})})|
     + \sum_{i=2}^N
    \int Q_{\adaptstate'}(\refr{x}_1,\ud x_0)
    | Q_\adaptstate(x_0, g_i^{(\refr{x}_{1:T},x_0)} ) -
    Q_{\adaptstate'}(x_0, g_i^{(\refr{x}_{1:T},x_0)} ) |
\end{align*}

with functions defined below, which satisfy
$\|g_0^{(\refr{x}_{1:T})}\|_\infty\le 1$ and
$\|g_i^{(\refr{x}_{1:T},x_0)}\|_\infty \le 1$:
\begin{align*}
    g_{0}^{(\refr{x}_{1:T})}(x_0) &= \int \delta_{\refr{x}_1}(\ud
    \tilde{x}_1^{(1)})\prod_{k=2}^N
    Q_\adaptstate(x_0,\ud \tilde{x}_1^{(k)})
    P_{\mathrm{CPF}}(\tilde{x}_{1}^{(1:N)}, \refr{x}_{2:T}, f), \\
    g_i^{(\refr{x}_{1:T},x_0)}(\tilde{x}_1^{(i)}) &= \delta_{\refr{x}_1}(\ud
    \tilde{x}_1^{(1)}) \prod_{k=2}^{i-1}
    Q_{\adaptstate'}(x_0,\ud \tilde{x}_1^{(k)})
    \prod_{k=i+1}^N Q_\adaptstate(x_0,\ud \tilde{x}_1^{(k)})
    P_{\mathrm{CPF}}(\tilde{x}_{1}^{(1:N)}, \refr{x}_{2:T}, f).
    \qedhere
\end{align*}
\end{proof}

The following result is direct:
\begin{lemma}
\label{lemma:tv_bound} 
Let $Q_\Sigma$ stand for the random-walk Metropolis type kernel with
increment proposal distribution $q_\Sigma$, and with target function $M_1\ge 0$,
that is, a transition probability of the form:
\[
    Q_\Sigma(x,A) = \int_A q_\Sigma(\ud z) \min\bigg\{1,
      \frac{M_1(x+z)}{M_1(x)}\bigg\} + 1(x\in A)\bigg(1 - \int
    q_\Sigma(\ud z) \min\bigg\{1,
      \frac{M_1(x+z)}{M_1(x)}\bigg\}\bigg).
\]
Then, $\|Q_\Sigma(x,\uarg) - Q_{\Sigma'}(x,\uarg)\|_\mathrm{tv} \le 2 \|
q_\Sigma - q_{\Sigma'}\|_\mathrm{tv}$.
\end{lemma}

The following result is from the proof of \cite[Proposition 26]{vihola-asm}:
\begin{lemma}
    \label{lem:proposal-continuity} 
Let $q_\Sigma(x,\ud y)$ stand for the centred Gaussian distribution
with covariance $\Sigma$, or the centred multivariate $t$-distribution
with shape $\Sigma$ and some constant degrees of freedom $\nu>0$.
Then, for any
$0<b_\ell<b_u<\infty$ there
exists a constant $c=c(b_\ell,b_u)<\infty$ such that
for all $\Sigma,\Sigma'$ with all eigenvalues within $[b_\ell,b_u]$,
\[
    \| q_\Sigma - q_{\Sigma'}\|_\mathrm{tv}
    \le c \| \Sigma - \Sigma' \|,
\]
where the latter stands for the Frobenius norm in $\mathbb{R}^d$.
\end{lemma}

\begin{assumption}[Mixing]
    \label{a:mixing} 
The collection of Markov transition probabilities
$\{Q_\adaptstate\}_{\adaptstate\in\adaptspace}$
on $\X$ satisfies:
\begin{enumerate}[(i)]
    \item \label{item:minorisation} There exists $\epsilon>0$ such that
      for each $\adaptstate$, there is a probability measure
      $\nu_{\adaptstate}$ satisfying
      $Q_\adaptstate(x_0,A) \ge \epsilon \nu_{\adaptstate}(A)$ for all
      $x_0\in \mathsf{X}$, $\adaptstate\in\adaptspace$ and measurable $A\subset\mathsf{X}$.
    \item \label{item:bounded-weights} $\|G_k\|_\infty<\infty$ for all $k=1,\ldots,T$.
    \item \label{item:normalising-constant}
      $\inf_{\adaptstate\in\mathsf{Z}} \int \nu_{\adaptstate}(\ud x_0) Q_\adaptstate(x_0, \ud x_1) G_1(x_1) \prod_{k=2}^T
    M_k(x_{k-1}, \ud x_k) G_k(x_{k-1},x_k) \ud x_{1:T} > 0$.
\end{enumerate}
\end{assumption}

\begin{lemma}
    \label{lem:minorisation} 
Suppose that Assumption \ref{a:mixing} holds, then the kernels
$P_{\adaptstate}$ and $\tilde{P}_\adaptstate$ satisfy simultaneous
minorisation conditions, that is, there exists $\delta>0$ and probability
measures $\nu_\adaptstate,\tilde{\nu}_\adaptstate$, such that
\begin{equation}
     P_{\adaptstate}^k(x_{1:T},\uarg)
     \ge \delta \nu_\adaptstate(\uarg)
     \qquad\text{and}\qquad
      \tilde{P}_{\adaptstate}^k(\tilde{x}_{1:T},\uarg) \ge \delta
      \tilde{\nu}_\adaptstate(\uarg),
    \label{eq:uniform-ergodicity}
\end{equation}
for all $x_{1:T}\in\X$, $\tilde{x}_1^{(1:N)}\in\X^N$, and
$\adaptstate\in\adaptspace$.
\end{lemma}
\begin{proof} 
For $\hat{P}_\adaptstate \in \{P_\adaptstate, \tilde{P}_\adaptstate\}$,
we may write as in the proof of Lemma \ref{lem:kernel-diff}
\[
    \hat{P}_{\adaptstate}(x_{1:T},\uarg) = \int Q_\adaptstate(x_1, \ud x_0)
\hat{P}^*_{\text{CPF},\adaptstate,x_0}(x_{1:T},\uarg),
\]
where the latter term refers to the term in brackets in
\eqref{eq:ai-cpf-decomposition} --- the transition
probability of a conditional particle filter, with reference
$x_{1:T}$, and the
Feynman-Kac model $\check{M}_1^{(\adaptstate,x_0)}(\ud x_1) =
Q_\adaptstate(x_0, \ud x_1)$,
$M_{2:T}$ and $G_{1:T}$, whose normalised probability we call
$\pi^*_{\adaptstate,x_0}$.
Assumption \ref{a:mixing}
\ref{item:bounded-weights} and \ref{item:normalising-constant}
guarantee that $P^*_{\text{CPF},\adaptstate,x_0}(x_{1:T},\ud x'_{1:T}) \ge
\varepsilon \pi^*_{\adaptstate,x_0}(\ud x'_{1:T})$,
where $\hat{\epsilon}>0$ is independent of $x_0$ and $\adaptstate$
\cite[Corollary 12]{andrieu-lee-vihola}. Note that the same conclusion
holds also with backward sampling, because it is only a further Gibbs step to the
standard CPF. Likewise, in case of $\tilde{P}_\adaptstate$,
the result holds because we may regard $\tilde{P}_\adaptstate$ as an
augmented version of $P_\adaptstate$ \citep[e.g.][]{franks-vihola}.
We conclude that
\[
    \hat{P}_{\adaptstate}(x_{1:T},\uarg)
    \ge \epsilon \hat{\epsilon} \int \nu_\adaptstate(\ud x_0)
     \pi^*_{\adaptstate,x_0}(\uarg),
\]
where the integral defines a probability measure independent of
$x_{1:T}$.
\end{proof}

We may write the $k$:th step of Algorithm \ref{alg:aai-cpf} as:
\begin{enumerate}[(i)]
\item $(X_k,\adaptdata_k) \sim
  \tilde{P}_{\adaptstate_{k-1}}(X_{k-1},\uarg)$,
\item $\adaptstate_k^* = \adaptstate_{k-1} + \eta_k
  H(\adaptstate_{k-1}, X_k, \adaptdata_k )$,
\end{enumerate}
where $H$ correspond to Algorithm
\ref{alg:adapt-fdi-am} or \ref{alg:adapt-fdi-aswam}, respectively.
The stability may be enforced by introducing the following optional
step:
\begin{enumerate}[(i)]
    \stepcounter{enumi}
    \stepcounter{enumi}
\item $\adaptstate_k = \adaptstate_k^* 1(\adaptstate_k\in \adaptspace) +
  \adaptstate_{k-1} 1(\adaptstate_k^*\notin \adaptspace)$,
\end{enumerate}
which ensures that $\adaptstate\in \adaptspace$, the
feasible set for adaptation.

\begin{proof}[Proof of Theorem \ref{thm:slln}] 
The result follows by \cite[Theorem 2]{saksman-vihola},
as (A1) is direct,
Lemma \ref{lem:minorisation} implies (A2)
with $V\equiv 1$, $\lambda_n=0$, $b_n=1$, $\delta_n=\delta$
and $\epsilon=0$, Lemmas \ref{lemma:tv_bound} and
Lemma \ref{lem:proposal-continuity} imply (A3), and (A4) holds trivially,
as $\| H (\uarg)\|_\infty < \infty$, thanks to the compactness of $D$.
\end{proof}


\section{Details of the DPG-BS algorithm} \label{app:dpg_details} 
The diffuse particle Gibbs algorithm targets \eqref{eq:feynman_kac} by alternating
the sampling of $x_{2:T}$ given $x_1$, and $x_1$ given $x_{2:T}$.
Hence, the algorithm is simply particle Gibbs where the initial state is treated
as a parameter.
Define
$$p^{\mathrm{DPG}}(x_1  \mid x_{2:T}) \propto M_1(x_1)G_1(x_1)M_2(x_2 \mid x_1)G_2(x_1, x_2).$$
With this definition, the DPG-BS algorithm can be written as in Algorithm \ref{alg:dpg}.
Lines \ref{line:dpg-cpf-latent-start}--\ref{line:dpg-cpf-latent-end} constitute a CPF-BS update for
$x_{2:T}$, and line \ref{line:dpg-cpf-x1-update} updates $x_1$.
A version of the RAM algorithm \cite{ram} (Algorithm \ref{alg:ram}) is used for adapting
the normal proposal used in sampling $x_1$ from $p^{\mathrm{DPG}}$.

\begin{algorithm}
  \caption{$\text{DPG-BS}(X_1^{(0)}, \refr{x}_{2:T}^{(0)}; \pi, N$)}
  \label{alg:dpg}
  \begin{algorithmic}[1]
    \State Set $S_0 = I, \alphatarget = 0.441, \eta^{\mathrm{max}} = 0.5, \gamma = 0.66.$
    \For{$j$ in $1, \dots, n$}
      \State Simulate $\tilde{X}_2^{(2:N)} \sim M_2(\uarg \mid X_1^{(j-1)})$ and set $\tilde{X}_2^{(1)} = \refr{x}_2^{(0)}$. \label{line:dpg-cpf-latent-start}
      \State $(\tilde{X}_{2:T}^{(1:N)},W_{2:T}^{(1:N)},A_{2:T-1}^{(1:N)}) \gets \text{F-CPF}(\refr{x}_{3:T},
      \tilde{X}_2^{(1:N)}; M_{3:T}, G_{2:T}, N) $.

      \State $(B_{2:T}, \adaptdata) \gets
      \textsc{PickPath-BS}(\tilde{X}_{2:T}^{(1:N)},W_{2:T}^{(1:N)},
      A_{2:T-1}^{(1:N)}, M_{3:T}, G_{3:T})$ \label{line:dpg-cpf-latent-end}

      \State Simulate $(X_1^{(j)}, S_j) \gets \mathrm{RAM}(p^{\mathrm{DPG}}(\uarg \mid \tilde{X}_{2:T}^{(B_{2:T})}), X_1^{(j-1)}, S_{j-1}, \alphatarget, \eta^{\mathrm{max}}, \gamma).$ \label{line:dpg-cpf-x1-update}

      \State Set $\mathbf{X}^{(j)} = (X_1^{(j)}, \tilde{X}_{2}^{(B_2)}, \tilde{X}_{3}^{(B_3)},
                  \ldots, \tilde{X}_{T}^{(B_T)}).$
    \EndFor
    \State \textbf{output} $\mathbf{X}^{(1:n)}$
  \end{algorithmic}
\end{algorithm}

\begin{algorithm}
  \caption{$\text{RAM}(p, \theta^{(n-1)}, S_{n-1}, \alpha_{*}, \eta^{\mathrm{max}}, \gamma)$ (iteration $n$)}
  \label{alg:ram}
  \begin{algorithmic}[1]
    \State Simulate $U_n \sim N(0, I_d)$.
    \State Propose $\theta^{*} = \theta^{(n-1)} + S_{n-1}U_n$.
    \State Compute $\alpha_n = \min\left\{ 1, \dfrac{p(\theta^*)}{p(\theta^{(n-1)})} \right\}$.
    \State With probability $\alpha_n$, set $\theta^{(n)} = \theta^{*}$; otherwise
           set $\theta^{(n)} = \theta^{(n-1)}$.
    \State Set $\eta_n = \min\{\eta^{\mathrm{max}}, dn^{-\gamma} \}.$
    \State Compute $S_n$ such that $S_nS_n^{'} = S_{n-1}\left(I + \eta_n(\alpha_n - \alphatarget)
           \dfrac{U_nU_n^{'}}{\lVert U_n \rVert^2}\right)S_{n-1}^{'}.$
    \State \textbf{output} $\theta^{(n)}$, $S_n$.
  \end{algorithmic}
\end{algorithm}


\section{Supplementary Figures} 

\begin{figure}
  \centering
  \includegraphics{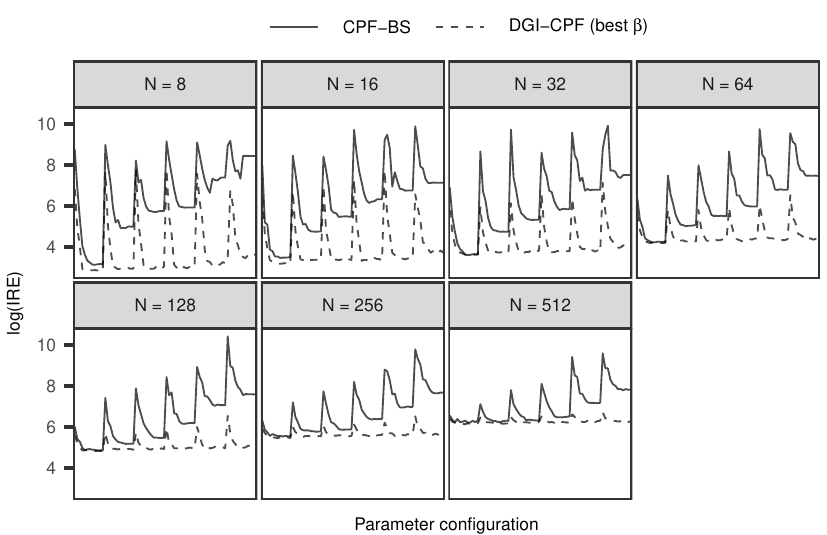}
  \caption{The $\log{(\ire)}$ resulting from the application of the CPF-BS and
  the best case DGI-CPF to the SV model.
  The horizontal axis depicts different configurations of
  $\sigma_1$ and $\sigma_x$, and in each panel $N$ varies.}
  \label{fig:dgi-cpf-opt-vs-cpf-bs-SV-logiactn}
\end{figure}

\begin{figure}
  \centering
  \includegraphics{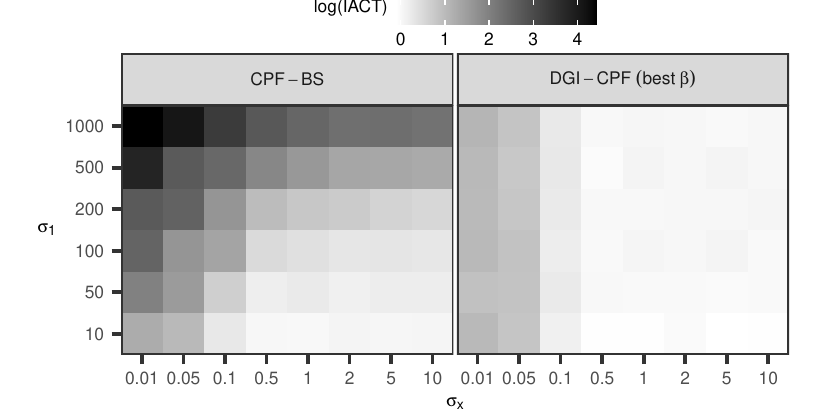}
  \caption{The $\log{(\iact)}$ of the CPF-BS (left) and the best case
  DGI-CPF (right) with respect to $\sigma_1$ and $\sigma_x$ in the case of the
  RW model and $N = 256$.}
  \label{fig:cpf-bs-dgi-cpf-NOISYAR-log-iact-heatmap-npar256}
\end{figure}

\begin{figure}
  \centering
  \includegraphics{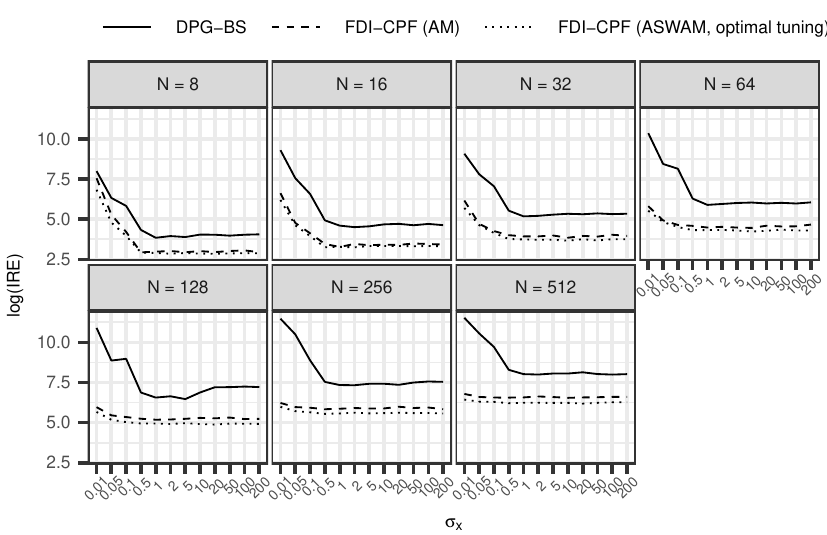}
  \caption{
  The $\log{(\ire)}$ for the DPG-BS, the FDI-CPF with the AM adaptation and the best case
  FDI-CPF with the ASWAM adaptation to the datasets generated with varying $\sigma_x$ from
  the SV model.}
  \label{fig:SV-fdi-cpf-vs-dpg-cpf-logiactn}
\end{figure}

\begin{figure}
  \centering
  \includegraphics{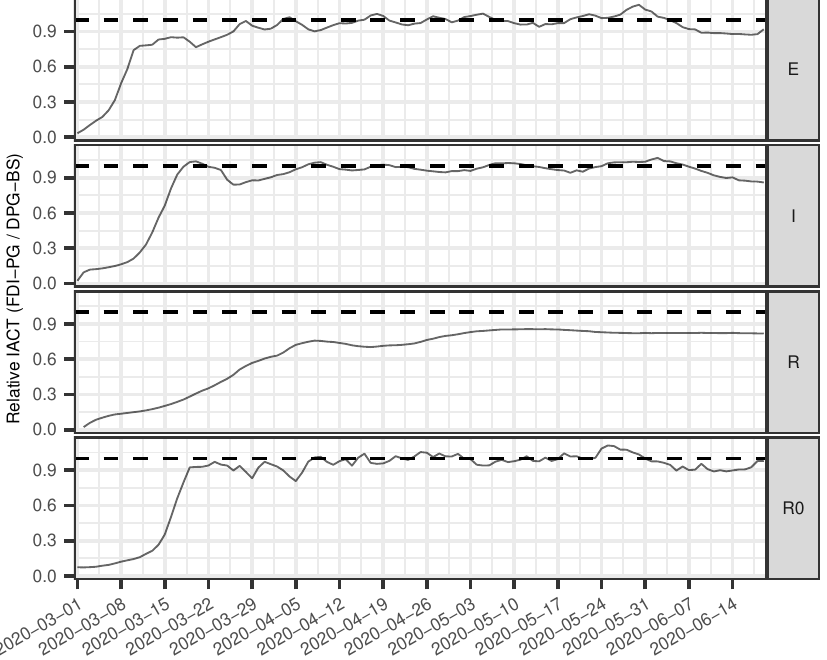}
  \caption{The integrated autocorrelation time with the FDI-PG relative to that of the DPG-BS
           for the state variables at each time point in the SEIR model. The dashed line shows the
           line of equal sampling efficiency. The first value for the state variable
           $R$ is missing, since $R_1 = 0$ is assumed in the model.}
  \label{fig:seir-rel-iact-fdi-dpg}
\end{figure}

\begin{figure}
  \centering
  \includegraphics{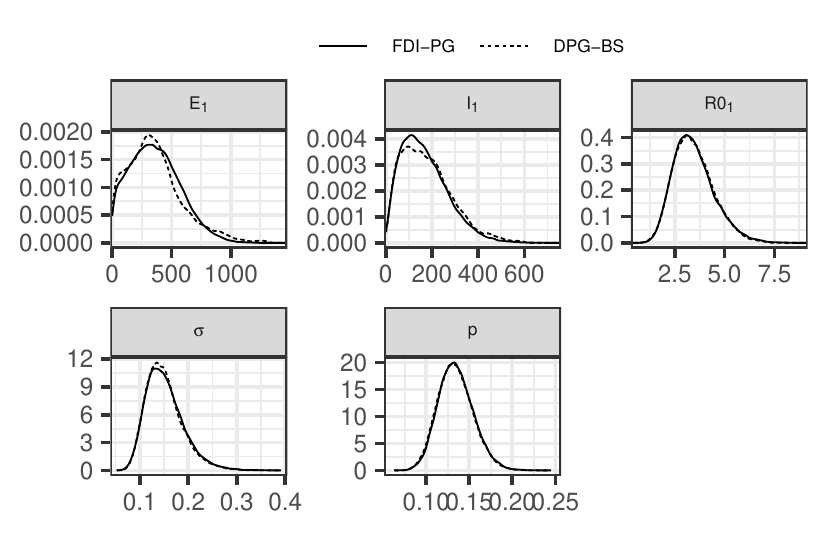}
  \caption{Marginal density estimates of the initial states and
  parameters for the SEIR model computed by the FDI-PG and DPG-BS.}
  \label{fig:seir-dens-per-method-and-param}
\end{figure}

\end{document}